\documentclass[11pt]{amsart}
\pdfoutput=1

\usepackage{amssymb,amsmath,amsthm}
\usepackage{latexsym}
\usepackage{arydshln}
\usepackage{graphicx}
\usepackage{tikz-cd}
\usepackage{tikz}
\usetikzlibrary{positioning}
\usetikzlibrary{arrows.meta}
\usepackage{mathtools,amssymb}
\usepackage{fancyhdr}
\usepackage{caption} 
\usepackage{array}
\usepackage{multirow, hhline}
\usepackage[ruled,vlined,linesnumbered]{algorithm2e}
\usepackage{url}
\usepackage{booktabs}
\usepackage{geometry}
\usepackage{subcaption}

\newtheorem{theorem}{Theorem}[section]

\newtheorem{lemma}[theorem]{Lemma}
\newtheorem{proposition}[theorem]{Proposition}
\theoremstyle{definition}
\newtheorem{definition}[theorem]{Definition}
\theoremstyle{remark}
\newtheorem{remark}[theorem]{Remark}
\numberwithin{equation}{section}
\newcommand{\Real}{\mathbb R}

\newcommand{\set}[1]{\left\{#1\right\}}

\begin{document}

\title{Denoising data reduction algorithm for Topological Data Analysis}

\author{Seonmi Choi}
\author{Semin Oh}
\author{Jeong Rye Park}
\author{Seung Yeop Yang}

\email[Seonmi Choi]{smchoi@seowon.ac.kr}
\email[Semin Oh]{semin@knu.ac.kr}
\email[Jeong Rye Park]{parkjr@gist.ac.kr}
\email[Seung Yeop Yang]{seungyeop.yang@knu.ac.kr}

\address{Department of Mathematics Education, Seowon University, Cheongju, 28674, Republic of Korea}
\address{KNU G-LAMP Research Center, KNU Institute of Basic Sciences, Kyungpook National University, Daegu, 41566, Republic of Korea}
\address{Department of Mathematical Sciences, Gwangju Institute of Science and Technology, Gwangju 61005, Republic of Korea}
\address{KNU G-LAMP Research Center, KNU Institute of Basic Sciences, Department of Mathematics, Kyungpook National University, Daegu, 41566, Republic of Korea}

\subjclass{}%

\begin{abstract}

Persistent homology is a central tool in topological data analysis, but its application to large and noisy datasets is often limited by computational cost and the presence of spurious topological features. Noise not only increases data size but also obscures the underlying structure of the data.

In this paper, we propose the Refined Characteristic Lattice Algorithm (RCLA), a grid-based method that integrates data reduction with threshold-based denoising in a single procedure. By incorporating a threshold parameter $k$, RCLA removes noise while preserving the essential structure of the data in a single pass. We further provide a theoretical guarantee by proving a stability theorem under a homogeneous Poisson noise model, which bounds the bottleneck distance between the persistence diagrams of the output and the underlying shape with high probability. In addition, we introduce an automatic parameter selection method based on nearest-neighbor statistics. Experimental results demonstrate that RCLA consistently outperforms existing methods, and its effectiveness is further validated on a 3D shape classification task.
\end{abstract}

\keywords{Topological data analysis, Persistent homology, Topology-preserving data reduction, Denoising, Stability Theorem, Poisson Point Process}

\subjclass[2020]{Primary: 55N31. Secondary: 62R40, 68T09.}

\maketitle


\section{Introduction}\label{Intro}


   
   



Persistent homology is a central tool in topological data analysis (TDA) that captures topological features of data shapes, including connected components, loops, higher-dimensional voids, and related structures.
It has found applications in diverse areas including genomics~\cite{RabBlu20}, shape analysis~\cite{LSLIVACC13}, and scientific data exploration. For general background on TDA, see~\cite{Car09, EdeHar10, Ghr08, ZomCar05}.
However, the direct computation of persistent homology via the Vietoris--Rips complex faces a fundamental scalability challenge: the number of simplices grows exponentially in the size of the input point cloud, making the computation prohibitively expensive for large datasets.
To address this issue, various preprocessing methods have been proposed, including witness complexes~\cite{SC04}, random projection~\cite{RVT14, She14}, linear-size approximation of the Vietoris--Rips filtration~\cite{She13}, subsampling~\cite{CFLMRW15}, and cluster-based data reduction~\cite{MSW20, MW19, MMW18}. See~\cite{OPTGH17} for a comprehensive survey.

In practice, observed data is often corrupted by background noise.
Such noise points not only unnecessarily increase the data size but also introduce spurious features in the persistence diagram, thereby obscuring the genuine topological structure of the underlying shape.
Existing data reduction methods do not distinguish between noise and signal, so noise persists in the reduced output.
On the other hand, denoising methods such as adaptive DBSCAN~\cite{ZZLLL24}, LDOF~\cite{ZHJ09}, and LUNAR~\cite{GHNN22} can remove noise, but they require a separate preprocessing step and offer no guarantee on how well the topological features are preserved after denoising.
From a practical standpoint, it is therefore desirable to develop a single-step method that simultaneously reduces data size and removes noise, without requiring manual parameter tuning.

In previous work~\cite{COPYY}, the authors introduced the Characteristic Lattice Algorithm (CLA), a grid-based data reduction method for TDA.
CLA partitions the ambient space into hypercubes of side length $\delta$ and selects a representative point from each non-empty hypercube.
The method is simple to implement and computationally efficient, and has been shown to be competitive with existing data reduction techniques in terms of both speed and accuracy.
However, since CLA selects a representative from every non-empty hypercube, including those containing only noise points, noise is inevitably carried over into the reduced output. This limitation highlights the need for a method that can simultaneously reduce data size while effectively eliminating noise, thus motivating the present work.\\

In this paper, we propose the Refined Characteristic Lattice Algorithm (RCLA), an enhanced variant of the CLA that incorporates a threshold parameter $k$ to discard hypercubes containing fewer than $k$ points. This refinement enables simultaneous data reduction and noise removal within a single pass of the algorithm.
Moreover, the grid structure underlying RCLA is naturally amenable to probabilistic analysis.
When background noise is modeled as a Homogeneous Poisson Point Process (HPPP), the noise count in each cell becomes an independent Poisson random variable, which allows us to derive quantitative guarantees on the topological fidelity of the output.
The contributions of this paper are as follows:
\begin{enumerate}
    \item We propose RCLA, a denoising data reduction algorithm that combines grid-based data reduction with threshold-based filtering in a single pass, eliminating the need for a separate denoising step (Section~\ref{CLA_RCLA}). 
    \item We prove a stability theorem showing that, under the HPPP noise model, the bottleneck distance between the persistence diagrams of the RCLA output and the underlying shape is at most $\sqrt{m}\,\delta$ with high probability, where the confidence level is determined explicitly by the noise intensity and the grid configuration (Section~\ref{StabilityRCLA}).
    \item We present a fully automatic procedure for selecting the parameters $(\delta, k)$ from the data, based on nearest-neighbor distance distributions and Bayesian estimation. All experiments in this paper use the default settings without manual tuning (Section~\ref{ImplementationRCLA}).
    \item We validate RCLA experimentally on synthetic data, demonstrating that it consistently outperforms CLA and three denoising algorithms (Adaptive DBSCAN, LDOF, and LUNAR) on the tested configurations (Section~\ref{Experiments}). We further apply RCLA to a 3D shape classification task, achieving approximately $94.5\%$ data compression and $99.88\%$ classification accuracy (Section~\ref{Implementation}).
\end{enumerate}

The paper is organized as follows.
Section~\ref{Preliminaries} reviews background on persistent homology, including the bottleneck distance and the stability theorem.
Section~\ref{RCLA} introduces the RCLA algorithm, presents an automatic parameter selection procedure, and proves a stability theorem for RCLA.
Section~\ref{Experiments} presents experimental comparisons on synthetic data.
Section~\ref{Implementation} demonstrates an application to 3D shape classification.
Section~\ref{Conclusion} provides concluding remarks.

\vspace{0.5cm}
\section{Preliminaries}\label{Preliminaries}

We recall the Vietoris–Rips complex, the persistent homology, and its representations via barcodes and persistence diagrams. 
We also introduce standard distances on persistence diagrams and state the stability theorem.
See \cite{EdeHar10, RabBlu20} for more details.\\


Let $(X, d_X)$ be a finite metric space and let $\varepsilon > 0$.  
The \emph{Vietoris--Rips complex} $VR_\varepsilon(X, d_X)$ is an abstract simplicial complex defined as follows:
\begin{itemize}
    \item[(1)] The vertex set of $VR_\varepsilon(X, d_X)$ is $X$ itself,
    \item[(2)] A subset $\{v_0, v_1, \dots, v_k\} \subseteq X$ forms a $k$-simplex in $VR_{\varepsilon}(X, d_X)$ if  
    \[
    d_X(v_i, v_j) \leq 2\varepsilon \quad \text{for all } 0 \leq i, j \leq k.
    \]
\end{itemize}

\medskip

Given $\varepsilon < \varepsilon'$, it holds that $VR_{\varepsilon}(X, d_X) \subseteq VR_{\varepsilon'}(X, d_X)$, and the inclusion of simplicial complexes induces a simplicial map
$VR_{\varepsilon}(X, d_X) \longrightarrow VR_{\varepsilon'}(X, d_X)$.
Since $X$ is finite, there exist finitely many critical values 
$\varepsilon_1 < \varepsilon_2 < \cdots < \varepsilon_m$
such that 
\[
VR_{\varepsilon_1}(X, d_X) \subsetneq VR_{\varepsilon_2}(X, d_X) \subsetneq \cdots \subsetneq VR_{\varepsilon_m}(X, d_X).
\]
The resulting filtration $\{VR_{\varepsilon_i}(X, d_X)\}_{i=1}^m$ is referred to as the \emph{filtered Vietoris--Rips complex} on $X$.
Applying the $k$th persistent homology functor with coefficients in a field $\mathbb{F}$ yields the induced homomorphisms
\[
PH_k(VR_{\varepsilon_1}(X, d_X); \mathbb{F}) \rightarrow PH_k(VR_{\varepsilon_2}(X, d_X); \mathbb{F}) \rightarrow \cdots \rightarrow PH_k(VR_{\varepsilon_m}(X, d_X); \mathbb{F}).
\]


For $\varepsilon_a<\varepsilon_b$, let $\theta_{a,b}\colon PH_k(VR_{\varepsilon_a}(X,d_X);\mathbb{F})\to PH_k(VR_{\varepsilon_b}(X,d_X);\mathbb{F})$
denote the homomorphism induced by the inclusion $VR_{\varepsilon_a}(X,d_X)\hookrightarrow VR_{\varepsilon_b}(X,d_X)$.
For an element $\gamma\in PH_k(VR_{\varepsilon_i}(X,d_X);\mathbb{F})$, we say that $\gamma$
\begin{itemize}
    \item[(1)] is \emph{born at $\varepsilon_i$} if $\gamma\notin \mathrm{Im}(\theta_{j,i})$ for all $\varepsilon_j<\varepsilon_i$,
    \item[(2)] \emph{dies at $\varepsilon_l>\varepsilon_i$} if either $\theta_{i,l}(\gamma)=0$ in $PH_k(VR_{\varepsilon_l}(X,d_X);\mathbb{F})$, or there exist $\varepsilon_j<\varepsilon_l$ and $\alpha\in PH_k(VR_{\varepsilon_j}(X,d_X);\mathbb{F})$ such that
$\theta_{i,l}(\gamma)=\theta_{j,l}(\alpha)$.
\end{itemize}
The half-open interval $[\varepsilon_i,\varepsilon_l)$ is called the \emph{bar} associated to $\gamma$, and the multiset of all such intervals in degree $k$ is the $k$th \emph{barcode}, denoted by $B_k(X)$.
Bars with large length $\varepsilon_l-\varepsilon_i$ represent persistent topological features, whereas short bars are typically interpreted as noise.

By associating to each bar $[a,b)\in B_k(X)$ the point $(a,b)\in\mathbb{R}^2$, we obtain an equivalent representation of the barcode, called the \emph{persistence diagram} (PD). It is denoted by $PD_k(X)$.

To quantify the effect of perturbations on persistent homology, we compare persistence diagrams via the bottleneck distance. 
The stability theorem is formulated with respect to this distance.

For two non-empty subsets $C$ and $D$ of a metric space $(X, d_{X})$, define the \emph{Hausdorff distance} between $C$ and $D$ as
\[
d_H(C, D) = \inf \{ \varepsilon > 0 \mid D \subseteq C_{\varepsilon}, \; C \subseteq D_{\varepsilon} \},
\]
where $C_{\varepsilon}$ and $D_{\varepsilon}$ denote the $\varepsilon$-neighborhoods of $C$ and $D$, respectively.

\begin{definition}
The \emph{Gromov--Hausdorff distance} between two compact metric spaces $(X, d_X)$ and $(Y, d_Y)$ is given by  
\[
d_{GH} \left((X, d_X), (Y, d_Y)\right) = \inf_{\substack{\theta_1: X \to Z \\ \theta_2: Y \to Z}} d_H \left( \theta_1(X), \theta_2(Y) \right),
\]
where $\theta_1$ and $\theta_2$ are isometric embeddings of $(X,d_X)$ and $(Y,d_Y)$, respectively, into a common metric space $(Z,d_Z)$.
\end{definition}

For intervals $[a_1,b_1)$ and $[a_2,b_2)$, define
\[
d_\infty \left([a_1,b_1),[a_2,b_2) \right)=\max\{|a_1-a_2|,\ |b_1-b_2|\},
\qquad 
d_\infty \left([a_1,b_1),\varnothing \right)=\frac{b_1-a_1}{2}.
\]
\begin{definition}
Let $B_1$ and $B_2$ be barcodes. 
The \emph{bottleneck distance} between $B_1$ and $B_2$ is defined by
\[
d_B(B_1,B_2)=\inf_{\phi}\ \sup_{I\in B_1} d_\infty\bigl(I,\phi(I)\bigr),
\]
where a matching is a bijection $\phi : C_1 \to C_2$ between sub-multisets $C_1\subseteq B_1$ and $C_2\subseteq B_2$ such that any interval not in $C_i$ is paired with $\varnothing$.
\end{definition}



\begin{proposition}\cite{CC-SGMO09, C-SEH07}\label{PropBD}
Let $(X,d_{X})$ and $(Y,d_{Y})$ be finite metric spaces. Then we have
\[
d_{B}(B_{k}(X), B_{k}(Y)) \leq d_{GH}((X,d_{X}), (Y,d_{Y}))
\]
for all $k\geq 0$.
\end{proposition}

\vspace{0.5cm}
\section{The Improved Characteristic Lattice Algorithm}\label{RCLA}

In this section, we improve the data reduction algorithm for TDA, called a {\it Characteristic Lattice Algorithm}, abbreviated as CLA, which was introduced in \cite{COPYY}.

\subsection{\bf{CLA and its improvement}}\label{CLA_RCLA}

Characteristic Lattice Algorithm is one of preprocessing algorithms for TDA to maintain topological features and decrease the computation time. 

The CLA process is designed to organize hypercubes according to a specified lattice size $\delta >0$. 
In the case that each hypercube contains elements in a given point cloud $X$ in $\mathbb{R}^{m}$, the process selects a sample point within the hypercube and replaces all existing points in $X$ within the hypercube with the selected sample point. 


The algorithm ensures that the characteristics of the given data can be analyzed with reduced computation time while preserving the topological features. 
Although this is an appropriate preprocessing method, it would be preferable to compute persistent homology in a more efficient manner by removing noisy data. 

Let $k$ be a positive integer. 
For the resulting data $X_{\delta}^{*}$ obtained from the data $X$ by applying CLA with $\delta$, 
if a hypercube contains fewer than or equal to $k$ elements in $X$, then a sample point will not be selected from that hypercube. 
Consequently, we can derive the new dataset, denoted as $X_{\delta, k}^{*}$, which is obtained by subtracting the points within that hypercubes from $X_{\delta}^{*}$.

This data reduction technique, which is an upgrade to the original CLA with additional processing, is called the {\it Refined Characteristic Lattice Algorithm}, abbreviated as RCLA, {\it with $\delta$ and $k$}.

\begin{figure}[h!]
  \centering
  \includegraphics[width = 12cm]{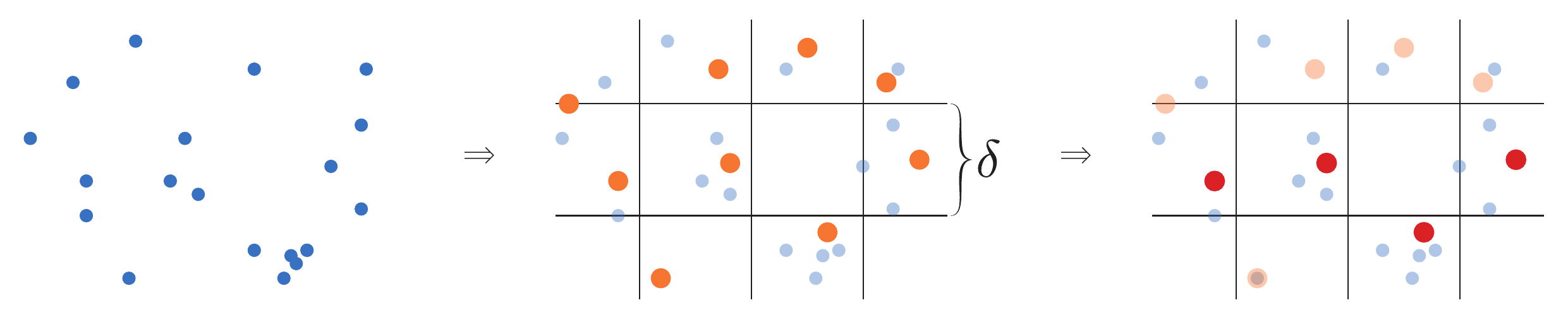}
  \caption{{\bf The steps of the CLA and RCLA algorithms}}\label{ExaRCLA1}
\end{figure}

\begin{remark}
In the construction of $X_{\delta,k}^*$, instead of selecting an arbitrary representative point from each selected cube, we may take the cube center. 
For every $m$-cube $C$ of side length $\delta$ with $|C\cap X|\ge k$, let $c(C)$ be the center of $C$. 
Define the resulting center-based set
\[
X_{\delta,k}^{c}:=\{\,c(C)\mid C \text{ is an $m$-cube of side } \delta,\ |C\cap X|\ge k\,\}.
\]
If $C=\prod_{i=1}^m [j_i\delta,(j_i+1)\delta)$, then
$c(C)=\bigl((j_1+\tfrac12)\delta,\dots,(j_m+\tfrac12)\delta\bigr).$
\end{remark}

The performance of CLA in the previous work has already been compared with several other preprocessing methods in terms of computational speed and accuracy (see \cite{COPYY} for more details). 
RCLA addresses the issue of noise in the data, thereby reducing the computation of persistent homology and enhancing the analysis of the phase space of the data. 
Therefore, for this comparative analysis, we will evaluate the efficacy of RCLA approach in comparison to the original CLA.
When implementing CLA in previous studies, all sample points were selected as the center of the hypercube for convenience. 
Similarly, the RCLA also use the center point for sample point selection.

\begin{figure}[h!]
  \centering
  \includegraphics[width = 3cm]{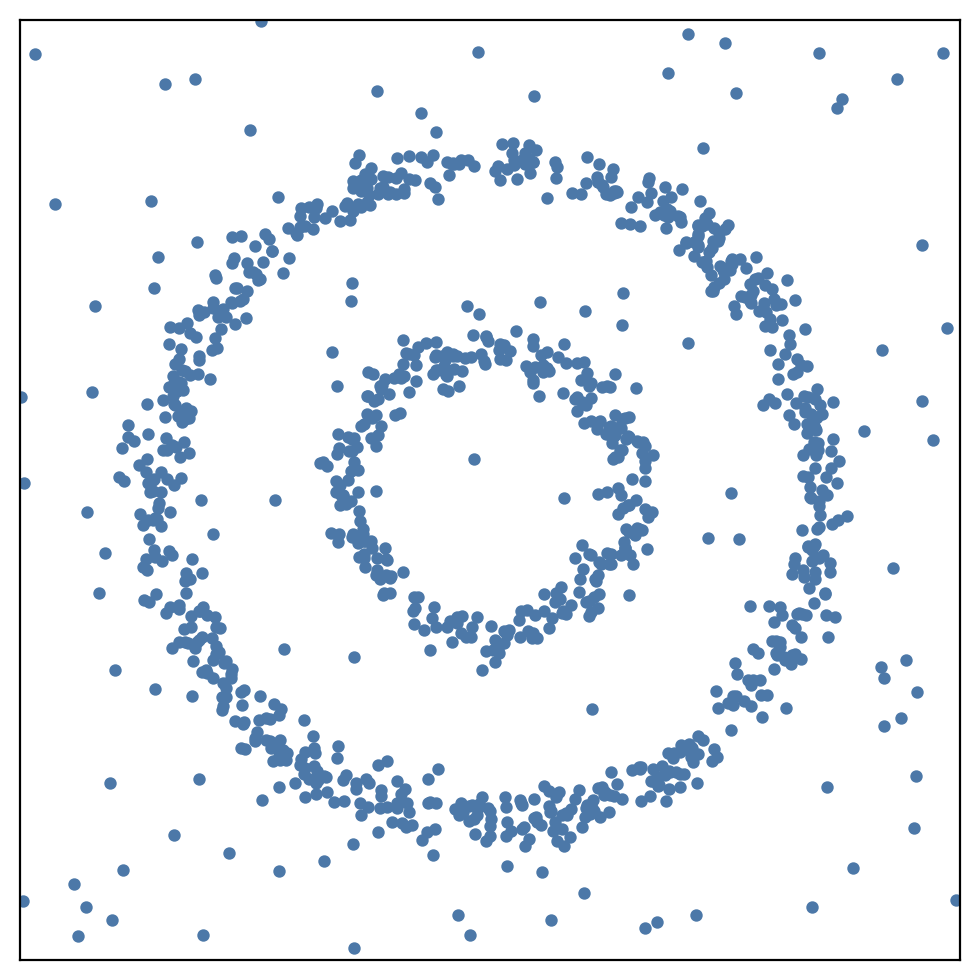}
  \includegraphics[width = 3cm]{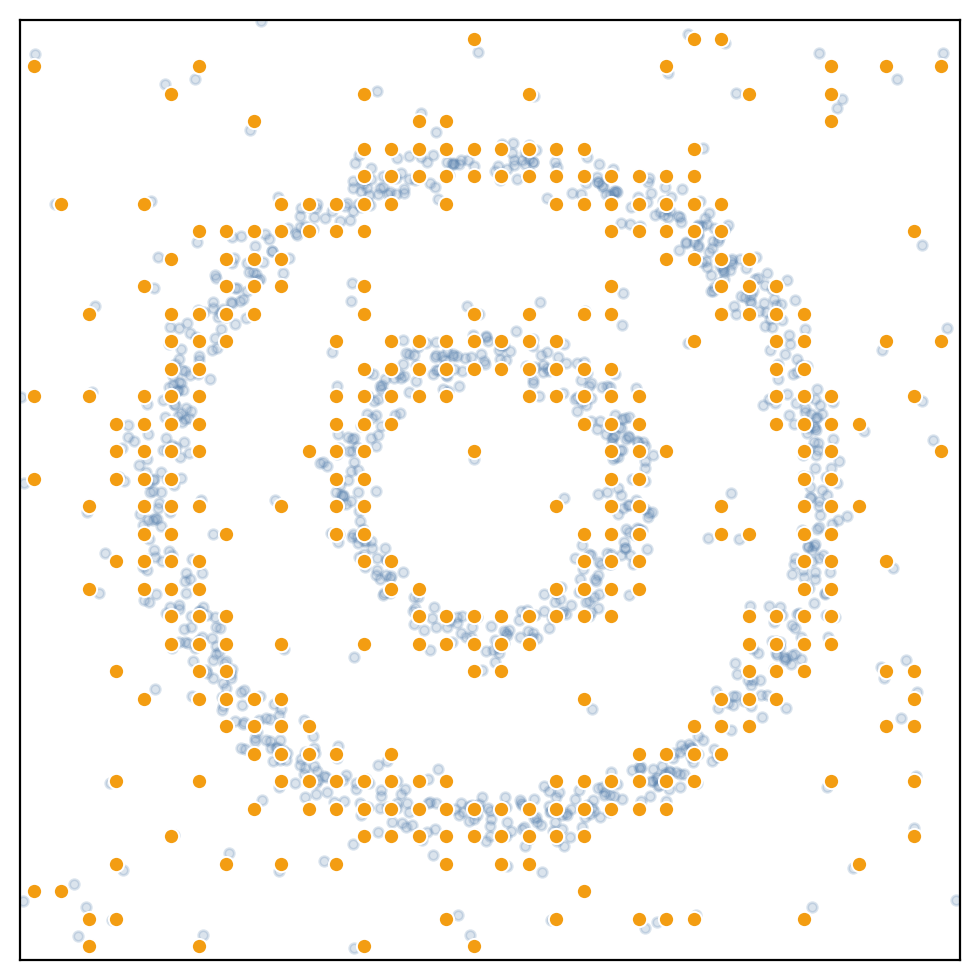}
  \includegraphics[width = 3cm]{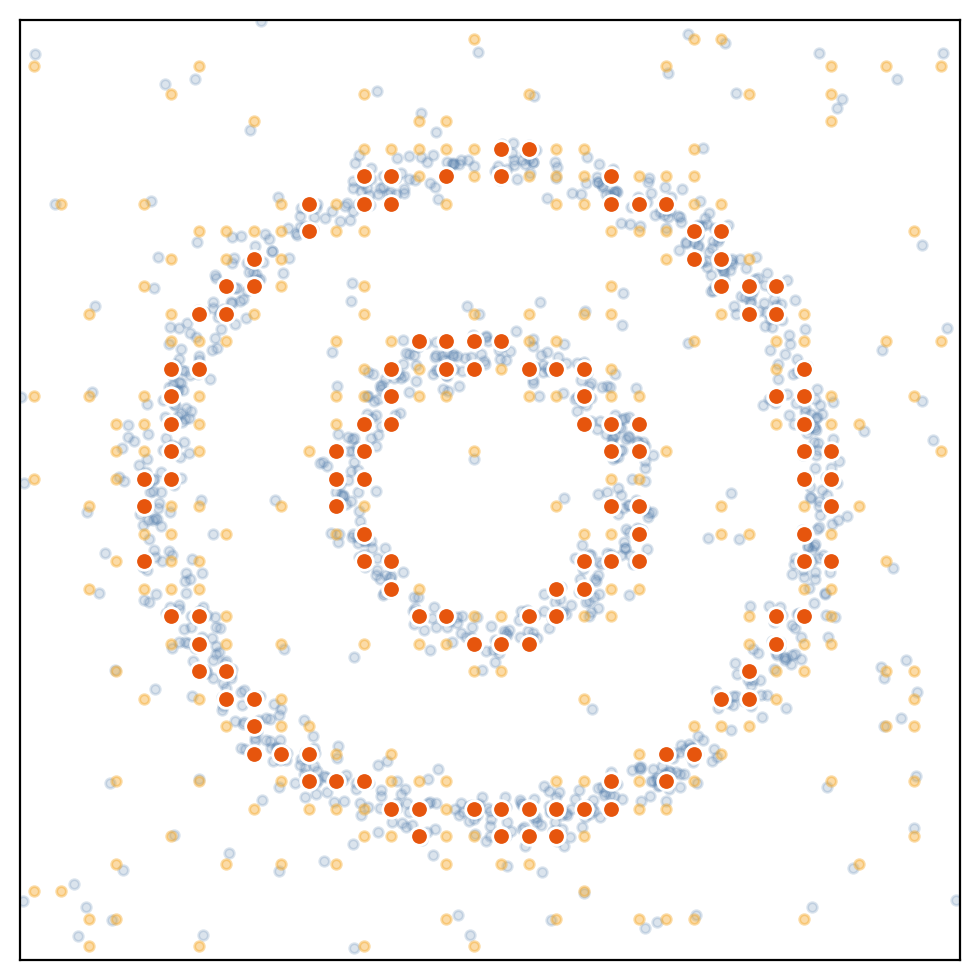}
  \caption{{\bf Datasets preprocessed by CLA and RCLA}}\label{ExaCLA1}
\end{figure}


\subsection{Stability theorem for RCLA}\label{StabilityRCLA}

Let $\mathcal{D}\subset\mathbb{R}^m$ be a bounded region (e.g., an $m$-dimensional rectangle).
We consider a point cloud $X$ in $\mathcal{D}$ consisting of shape data $X_{\mathrm{shape}}$ and noise data $X_{\mathrm{noise}}$, i.e., 
$$X=X_{\mathrm{shape}}\cup X_{\mathrm{noise}}\subset \mathcal{D}.$$
Assume that $X_{\mathrm{shape}}$ represents the underlying shape and that $X_{\mathrm{noise}}$ follows a Homogeneous Poisson Point Process (HPPP) with intensity $\lambda > 0$, independent of $X_{\mathrm{shape}}$.
The HPPP is a standard model for generating background noise uniformly over an ambient domain.
For more details on HPPP, see \cite{CSKM13, T03}. 

The purpose of applying RCLA to the observed point cloud $X$ is to obtain a reduced set that preserves the topological features of the underlying shape component $X_{\mathrm{shape}}$. 
In particular, for suitable choices of the parameters (e.g., the grid scale $\delta$ and the threshold $k$), we aim to show that the persistence diagram associated with the RCLA output is close to that of $X_{\mathrm{shape}}$ with respect to a standard metric on persistence diagrams (such as the bottleneck distance).

We fix a lattice scale $\delta>0$ and partition the ambient domain $\mathcal{D}\subset\mathbb{R}^m$ into (half-open) $m$-cubes of side length $\delta$. 
Let $\mathcal{C}=\{C_i\}_{i=1}^M$ denote the collection, so that
\begin{itemize}
    \item $C_i \cap C_j = \varnothing$ for $i \ne j$, 
    \item $\mathop{\cup}\limits_{i=1}^{M} C_i = \mathcal{D}$,
    \item $\mathrm{vol}(C_i) = \delta^m$ for each $i$.
\end{itemize}

The following proposition quantifies the discrepancy between these two choices in terms of the bottleneck distance.

\begin{theorem}
For a point cloud $X$, let $X^{\ast}_{\delta,k}$ be the set obtained by applying RCLA with $\delta$ and $k$. 
Let $X^{c}_{\delta,k}$ denote the corresponding center-based set obtained by selecting cube centers as representatives.
Then
\[
d_{B}(B_{n}(X^{\ast}_{\delta,k}), B_{n}(X^{c}_{\delta,k})) \leq \tfrac{\sqrt{m}\delta}{2}.
\]
\end{theorem}

\begin{proof}

When $\varepsilon = \frac{\sqrt{m}\delta}{2}$, it is obvious that $$X^{\ast}_{\delta, k} \subseteq (X^{c}_{\delta, k})_{\varepsilon} \quad \text{and} \quad X^{c}_{\delta, k}\subseteq(X^{\ast}_{\delta, k})_{\varepsilon}.$$
Then we have $d_H\bigl(X^{\ast}_{\delta,k},X^{c}_{\delta,k}\bigr)\le \tfrac{\sqrt{m}\delta}{2}.$
It follows that
\[
d_{GH}\bigl(X^{\ast}_{\delta,k},X^{c}_{\delta,k}\bigr)\le \tfrac{\sqrt{m}\delta}{2}
\]
by the definition of Gromov-Hausdorff distance.
Therefore, by Proposition \ref{PropBD}, 
\[
d_B(B_{n}(X^{\ast}_{\delta,k}), B_{n}(X^{c}_{\delta,k})) \leq \tfrac{\sqrt{m}\delta}{2}.
\]
\end{proof}

For each $m$-cube $C$, define
\[
N_{\mathrm{shape}}(C):=|C\cap X_{\mathrm{shape}}|,\qquad
N_{\mathrm{noise}}(C):=|C\cap X_{\mathrm{noise}}|.
\]
The Poisson assumption implies
\[
N_{\mathrm{noise}}(C)\sim \mathrm{Pois}(\mu)
\]
where $\mu=\lambda\,\delta^{m}$.
Note that $\{N_{\mathrm{noise}}(C_i)\}_{i=1}^M$ of $\mathcal{D}$ are independent.
The following proposition follows directly from equation~(2.11) in \cite{CSKM13}.

\begin{proposition}\label{prop:r}
  The probability that an $m$-cube $C$ contains at most $r$ points of $X_{\mathrm{noise}}$ is given by
  \[ 
  \mathbf{P}(N_{\mathrm{noise}}(C) \le r) = F_{\mathrm{Pois}(\mu)}(r) = e^{-\mu}\sum_{j=0}^{r} \frac{\mu^j}{j!},
  \] 
  where $F_{\mathrm{Pois}(\mu)}$ is the cumulative distribution function of the Poisson distribution with intensity $\mu$.  
\end{proposition}

We classify $m$-cubes according to their shape and noise points as follows. 

\begin{definition}
Let $C$ be an $m$-cube. We call $C$
\begin{itemize}
    \item a \emph{noise $m$-cube} if $N_{\mathrm{shape}}(C)=0$ and $N_{\mathrm{noise}}(C)\ge k$;
    \item a \emph{shape $m$-cube} if $N_{\mathrm{shape}}(C)>0$ and $|C\cap X|\ge k$;
    \item an \emph{out-shape $m$-cube} if $N_{\mathrm{shape}}(C)>0$ and $|C\cap X|< k$.
\end{itemize}
Then the sets of noise, shape, and out-shape $m$-cubes are denoted by $\mathcal{C}_{\mathrm{n}}$, $\mathcal{C}_{\mathrm{s}}$, and $\mathcal{C}_{\mathrm{o}}$, respectively.
\end{definition}

If $|C_i \cap X| \ge k$, then $C_i$ belongs to either $\mathcal{C}_{\mathrm{n}}$ or $\mathcal{C}_{\mathrm{s}}$.
If $N_{\text{shape}}(C_i) > 0$, then $C_i$ belongs to either $\mathcal{C}_{\mathrm{s}}$ or $\mathcal{C}_{\mathrm{o}}$.


\begin{lemma}\label{LemNoise}
The probability that there are no noise $m$-cubes is given by
\[
\mathbf{P}(\mathcal{C}_{\mathrm{n}} = \varnothing) = \left(F_{\mathrm{Pois}(\mu)}(k-1) \right)^{ |\{ C : N_{\mathrm{shape}}(C) = 0 \}|}.
\]
\end{lemma}

\begin{proof}
Note that every $m$-cube $C\in \mathcal{C}$ with $N_{\mathrm{shape}}(C)>0$ is not a noise $m$-cube.
For an $m$-cube $C$ with $N_{\mathrm{shape}}(C)=0$, $C\notin \mathcal{C}_{\mathrm{n}}$ if and only if $N_{\mathrm{noise}}(C)\leq k-1$.
By Proposition~\ref{prop:r},
\[
\mathbf{P}(N_{\text{noise}}(C) \le k-1) = F_{\mathrm{Pois}(\mu)}(k-1).
\]
Since $X_{\mathrm{noise}}$ is generated by an HPPP, the counts $\{N_{\mathrm{noise}}(C_i)\}_{i=1}^M$ are independent.
Therefore, we have 
\begin{align*}
\mathbf{P}(\mathcal{C}_{\mathrm{n}} = \varnothing) 
&= \prod_{C:\,N_{\mathrm{shape}}(C)=0}
\mathbf{P}\bigl(N_{\mathrm{noise}}(C)\leq k-1\bigr) \\
&= \left(F_{\mathrm{Pois}(\mu)}(k-1) \right)^{ |\{ C : N_{\text{shape}}(C) = 0 \}| }.
\end{align*}

\end{proof}

\begin{lemma}\label{LemOutShape}
The probability that there are no out-shape $m$-cubes is given by
\[
\mathbf{P}(\mathcal{C}_{\mathrm{o}}=\varnothing)=\prod_{C: N_{\mathrm{shape}}(C)>0}\Big(1-F_{\mathrm{Pois}(\mu)}(\max\{ 0, k-N_{\mathrm{shape}}(C)\} - 1)\Big).
\]
\end{lemma}

\begin{proof}
Note that every $m$-cube $C\in \mathcal{C}$ with $N_{\mathrm{shape}}(C)=0$ is not an out-shape $m$-cube.
For an $m$-cube $C$ with $N_{\mathrm{shape}}(C)>0$, $C\notin \mathcal{C}_{\mathrm{o}}$ if and only if $\lvert C \cap X \rvert \geq k$.

Since $\lvert C \cap X \rvert = N_{\mathrm{shape}}(C) + N_{\mathrm{noise}}(C)$, then 
$N_{\mathrm{noise}}(C) \ge \max\{ 0, k-N_{\mathrm{shape}}(C) \}.$
Put $r(C)=\max\{ 0, k-N_{\mathrm{shape}}(C) \}$.
Therefore  
\[
\mathbf{P}\big(|C \cap X|\ge k\big) =\mathbf{P}\big(N_{\mathrm{noise}}(C)\ge r(C)\big) =1-
\mathbf{P}\big(N_{\mathrm{noise}}(C)\le r(C)-1\big).
\]
By Proposition~\ref{prop:r}, 
\[
\mathbf{P}\big(|C \cap X|\ge k\big) = 1-F_{\mathrm{Pois}(\mu)}(r(C)-1).
\]
Here we use the convention
$F_{\mathrm{Pois}(\mu)}(-1)=0,$
so that, in the case $r(C)=0$, the corresponding factor is equal to $1$.
Since $\{N_{\mathrm{noise}}(C_i)\}_{i=1}^M$ are independent,
\[
\mathbf{P}(\mathcal{C}_{\mathrm{o}}=\varnothing) =\prod_{C:N_{\mathrm{shape}}(C)>0} \Big(1-F_{\mathrm{Pois}(\mu)}\big(r(C)-1\big)\Big).
\]
\end{proof}

Proposition~\ref{prop:r} together with the two Lemmas~\ref{LemNoise} and \ref{LemOutShape} provides the probabilistic estimates needed under the HPPP model to establish the stability theorem.


Define
\[
\alpha := 1 - \left(F_{\mathrm{Pois}(\mu)}(k-1)\right)^{|\{C:N_{\mathrm{shape}}(C) = 0\}|},
\]
\[
\beta := 1 - \prod_{C: N_{\mathrm{shape}}(C) > 0}  \big( 1 - F_{\mathrm{Pois(\mu)}}(r_C - 1) \big),
\]
where $\mu = \lambda \delta^m$ and $r(C) = \max\{0, k - N_{\mathrm{shape}}(C)\}$.

\begin{theorem}
Let $X=X_{\mathrm{shape}}\cup X_{\mathrm{noise}}$ be a point cloud on the bounded domain $\mathcal{D}\subset\mathbb{R}^m$. 
Assume that $X_{\mathrm{noise}}$ follows an HPPP  with intensity $\lambda>0$.
Fix $\delta>0$ and $k\in\mathbb{N}$, and let $X^{\ast}_{\delta,k}$ be the result set of RCLA.
Then, with probability at least $1-(\alpha + \beta)$, the following holds for all $n\geq 0$: 
\[
d_B\bigl(B_n(X_{\mathrm{shape}}),\,B_n(X^{\ast}_{\delta,k})\bigr)\le \sqrt{m} \delta.
\]
\end{theorem}

\begin{proof}
Let $E_{\mathrm{n}}:=\{\mathcal{C}_{\mathrm{n}}=\varnothing\}$ and 
$E_{\mathrm{o}}:=\{\mathcal{C}_{\mathrm{o}}=\varnothing\}$ be the events.
We bound the probabilities of these events separately.
\[
\mathbf{P}({E_{\mathrm{n}}})
=1-\mathbf{P}(\mathcal{C}_{\mathrm{n}}\neq\varnothing)= 1-\alpha
\quad \text{and} \quad 
\mathbf{P}({E_{\mathrm{o}}})
=1-\mathbf{P}(\mathcal{C}_{\mathrm{o}}\neq\varnothing)
= 1-\beta
\]
Therefore,
\[
\mathbf{P}(E_{\mathrm{n}}\cap E_{\mathrm{o}})
=1-\mathbf{P}(E_{\mathrm{n}}^c\cup E_{\mathrm{o}}^c)
\ge 1-\left(\mathbf{P}(E_{\mathrm{n}}^c)+\mathbf{P}(E_{\mathrm{o}}^c)\right)  
\ge 1-(\alpha+\beta).
\]
Applying this, we conclude that with probability at least $1-(\alpha+\beta)$, both $E_{\mathrm{n}}$ and $E_{\mathrm{o}}$ occur.
Hence, with probability at least $1-(\alpha+\beta)$, we have 
\[
\mathcal{C}_{\mathrm{n}}=\varnothing
\quad\text{and}\quad
\mathcal{C}_{\mathrm{o}}=\varnothing.
\]
Under this condition, for every point $x \in X_{\text{shape}}$, 
there is an $m$-cube $C$ obtained from $X^{\ast}_{\delta, k}$ such that $x\in C$.
When $\varepsilon = \sqrt{m} \delta$, it follows that
\[
X^{\ast}_{\delta, k} \subseteq (X_{\text{shape}})_{\varepsilon} 
\quad \text{and} \quad 
X_{\text{shape}} \subseteq (X^{\ast}_{\delta, k})_{\varepsilon}.
\]
Then $d_{H}(X_{\text{shape}}, X_{\delta,k}^{\ast}) \le \sqrt{ m }\delta$. 
By the definition of Gromov–Hausdorff Distance,
\[
d_{GH}(X_{\text{shape}}, X_{\delta,k}^{\ast}) \le \sqrt{ m }\delta.
\]
Hence, by Proposition~\ref{PropBD},
\[
d_B(B_{n}(X_{\text{shape}}), B_{n}(X_{\delta,k}^{\ast}))  \le \sqrt{ m }\delta.
\]
\end{proof}


\subsection{\bf{Implementation of RCLA}}\label{ImplementationRCLA}

By Section~\ref{StabilityRCLA}, the bottleneck distance between the persistence diagrams of the RCLA output and $X_{\mathrm{shape}}$ is bounded by $\sqrt{m}\,\delta$ for suitable $(\delta, k)$.
In practice, neither $X_{\mathrm{shape}}$ nor the noise intensity $\lambda$ is known, so $(\delta, k)$ must be selected from the data.
We describe a fully automatic procedure.

We generate candidate values of $\delta$ from the nearest-neighbor distance distribution of~$X$.
For each neighborhood order $q \in Q$ (default $Q = \{5, 8, 16\}$), we compute the $q$-th nearest-neighbor distance for every point in $X$ and pool all distances.
From the pooled distribution, we extract the quantile range $[q_{\mathrm{lo}}, q_{\mathrm{hi}}]$ (default $[0.01, 0.70]$) and generate $20$ candidate values on a logarithmic grid.

For each candidate $\delta$, we estimate a conservative upper bound on the noise intensity per cell.
Let $M$ denote the number of $\delta$-grid cells and $Z_0$ the number of empty cells.
Under the HPPP noise model, the probability that a pure-noise cell is empty is $p_0 = e^{-\mu}$, where $\mu = \lambda\,\delta^m$.
We estimate $p_0$ using the Jeffreys prior $\mathrm{Beta}(\tfrac{1}{2}, \tfrac{1}{2})$, a standard noninformative prior for binomial proportions~\cite{GCS13, Jef61}:
\begin{equation}\label{eq:posterior}
p_0 \mid Z_0 \sim \mathrm{Beta}\!\left(Z_0 + \tfrac{1}{2},\; M - Z_0 + \tfrac{1}{2}\right).
\end{equation}
Taking the $\gamma$-quantile $p_L$ (default $\gamma = 0.05$) gives a $(1-\gamma)$ credible lower bound on $p_0$, and since $-\ln$ is monotone decreasing, an upper bound on~$\mu$:
\begin{equation}\label{eq:mu_upper}
\mu_U(\delta) = -\ln p_L.
\end{equation}

Given $\mu_U(\delta)$, we choose the minimum $k$ such that the expected number of noise-only cells surviving the threshold is bounded:
\begin{equation}\label{eq:fp_budget}
M \cdot \mathbf{P}\!\left(\mathrm{Pois}(\mu_U) \ge k\right) \;\le\; \alpha_{\mathrm{fp}},
\end{equation}
where $\alpha_{\mathrm{fp}}$ (default $1.0$) bounds the per-family error rate~\cite{dunn1961multiple, lehmann2005generalizations}; a smaller $\alpha_{\mathrm{fp}}$ tightens the stability bound of Section~\ref{StabilityRCLA}.

For each candidate $\delta$, we run RCLA with the corresponding $k(\delta)$ to produce representative points $X^{c}_{\delta,k}$.
Candidates yielding fewer than $n_{\min}$ representative points (default $n_{\min} = 50$) are discarded.
Among the remaining candidates, we select the $\delta$ that minimizes
\begin{equation}\label{eq:quality_J}
J(\delta) = \sigma(d_{\mathrm{NN}}) \cdot \mu(d_{\mathrm{NN}}) + \eta \cdot (\beta_0 - 1),
\end{equation}
where $\mu(d_{\mathrm{NN}})$ and $\sigma(d_{\mathrm{NN}})$ are the mean and standard deviation of the 1-nearest-neighbor distances among the representative points, and $\beta_0$ is the number of connected components in a radius graph with radius $r = c_r \cdot \delta$ (default $c_r = 1.5$, $\eta = 1.0$).
The first term favors dense and uniformly spaced representatives; the second penalizes fragmentation.

\begin{remark}
The objective $J(\delta)$ is heuristic and is expected to exhibit a clear minimum: too fine a $\delta$ causes fragmentation, while too coarse a $\delta$ yields overly large spacing.
A potential failure mode arises with highly non-uniform density; in such cases, adjusting $n_{\min}$ can prevent the selection of an overly aggressive $\delta$ that discards meaningful structure.
\end{remark}

The complete procedure is summarized in Algorithm~\ref{alg:auto-select}.

\begin{algorithm}[htbp]
\caption{Automatic $(\delta, k)$ Selection}\label{alg:auto-select}
\KwData{Point cloud $X \in \Real^{n \times m}$; parameters $\alpha_{\mathrm{fp}},\, \eta,\, n_{\min}$}
\KwResult{Optimal parameters $(\delta^*, k^*)$}
$\Delta \leftarrow$ NN quantile-based $\delta$ candidates (log-space grid, $20$ values)\;
$J_{\mathrm{best}} \leftarrow \infty$\;
\ForEach{$\delta \in \Delta$}{
  $M \leftarrow$ total number of grid cells\;
  $Z_0 \leftarrow$ number of empty cells\;
  $\mu_U \leftarrow -\ln\!\left(\mathrm{Beta}(Z_0 + \tfrac{1}{2},\, M - Z_0 + \tfrac{1}{2}).\mathrm{quantile}(\gamma)\right)$\;
  $k \leftarrow \min\!\set{k : M \cdot \mathbf{P}(\mathrm{Pois}(\mu_U) \ge k) \le \alpha_{\mathrm{fp}}}$\;
  $C \leftarrow \mathrm{RCLA}(X, \delta, k)$\;
  \If{$|C| < n_{\min}$}{
    \textbf{continue}\;
  }
  $J(\delta) \leftarrow \sigma(d_{\mathrm{NN}}(C)) \cdot \mu(d_{\mathrm{NN}}(C)) + \eta \cdot (\beta_0(C) - 1)$\;
  \If{$J(\delta) < J_{\mathrm{best}}$}{
    $J_{\mathrm{best}} \leftarrow J(\delta)$\;
    $(\delta^*, k^*) \leftarrow (\delta, k)$\;
  }
}
\Return{$(\delta^*, k^*)$}\;
\end{algorithm}

Table~\ref{tab:hyperparams} lists the parameters that may be adjusted by the user.
All experiments in this paper use the defaults listed below without tuning.

\begin{table}[htbp]
\centering
\begin{tabular}{llcl}
\toprule
Parameter & Symbol & Default & Role \\
\midrule
False positive budget & $\alpha_{\mathrm{fp}}$ & $1.0$ & Max expected noise cells retained \\
Connectivity penalty & $\eta$ & $1.0$ & Weight on component count in~\eqref{eq:quality_J} \\
Radius factor & $c_r$ & $1.5$ & Radius $r = c_r \cdot \delta$ for $\beta_0$ \\
Minimum centers & $n_{\min}$ & $50$ & Discard if fewer representatives \\
\bottomrule
\end{tabular}
\caption{User-adjustable parameters for the automatic $(\delta, k)$ selection.}
\label{tab:hyperparams}
\end{table}

Algorithm~\ref{alg:auto-select} is not the computational bottleneck of the pipeline: it involves only sorting and nearest-neighbor queries on the input, repeated over $|\Delta|=20$ candidates.
The persistence computation on the compressed output is the dominant cost in practice.


\vspace{0.5cm}
\section{Experiments}\label{Experiments}


This section presents numerical results on a synthetic dataset for evaluating CLA and RCLA. 
In addition, we compare RCLA with several representative denoising methods to assess its ability to suppress background noise.

\subsection{\bf{Data}}\label{Data}

We consider a synthetic point cloud $X\subset \mathbb{R}^m$ of the form
\[
X = X_{\mathrm{shape}} \cup X_{\mathrm{noise}},
\]
where $X_{\mathrm{shape}}$ is a point sample from the underlying shape and $X_{\mathrm{noise}}$ consists of background noise points.

  

\begin{enumerate}
    \item
The first point cloud is generated for comparison of RCLA and CLA in Section~\ref{ComCLA}.
The shape component $X_{\mathrm{shape}}$ consists of points sampled from a circle embedded in $[0,1]^2$, while the noise component $X_{\mathrm{noise}}$ is generated from an HPPP in the same ambient region to model background noise. 
We fix $|X_{\mathrm{shape}}|=1000$ and set $|X_{\mathrm{noise}}| = r\,|X_{\mathrm{shape}}|$ where $r>0$ denotes the noise ratio.
Thus, the resulting dataset is $X = X_{\mathrm{shape}} \cup X_{\mathrm{noise}}$, as illustrated in the left data of Figure~\ref{ExaData} with $r = 0.1$.

    \item
The second point cloud is constructed for the experiments in Section~\ref{ComDenoising}, where RCLA is compared with established denoising methods.
We take $X_{\mathrm{shape}}$ as a point sample from a two-circle configuration in $[0,1]^2$, and we generate $X_{\mathrm{noise}}$ from a HPPP in the same ambient region. 
We fix $|X_{\mathrm{shape}}|=1000$, allocating $85\%$ of the points to the larger circle and the remaining $15\%$ to the smaller circle.
We set $|X_{\mathrm{noise}}| = r\,|X_{\mathrm{shape}}|$, where $r>0$ denotes the noise ratio. 
Then we obtain $X= X_{\mathrm{shape}} \cup X_{\mathrm{noise}}$.
The right data of Figure~\ref{ExaData} depicts the point cloud with $r=0.1$.
\end{enumerate}

\begin{figure}[h!]
  \renewcommand{\thesubfigure}{\arabic{subfigure}}
  \centering
  \begin{subfigure}[b]{0.4\textwidth}
    \centering
    \includegraphics[width = 3.5cm]{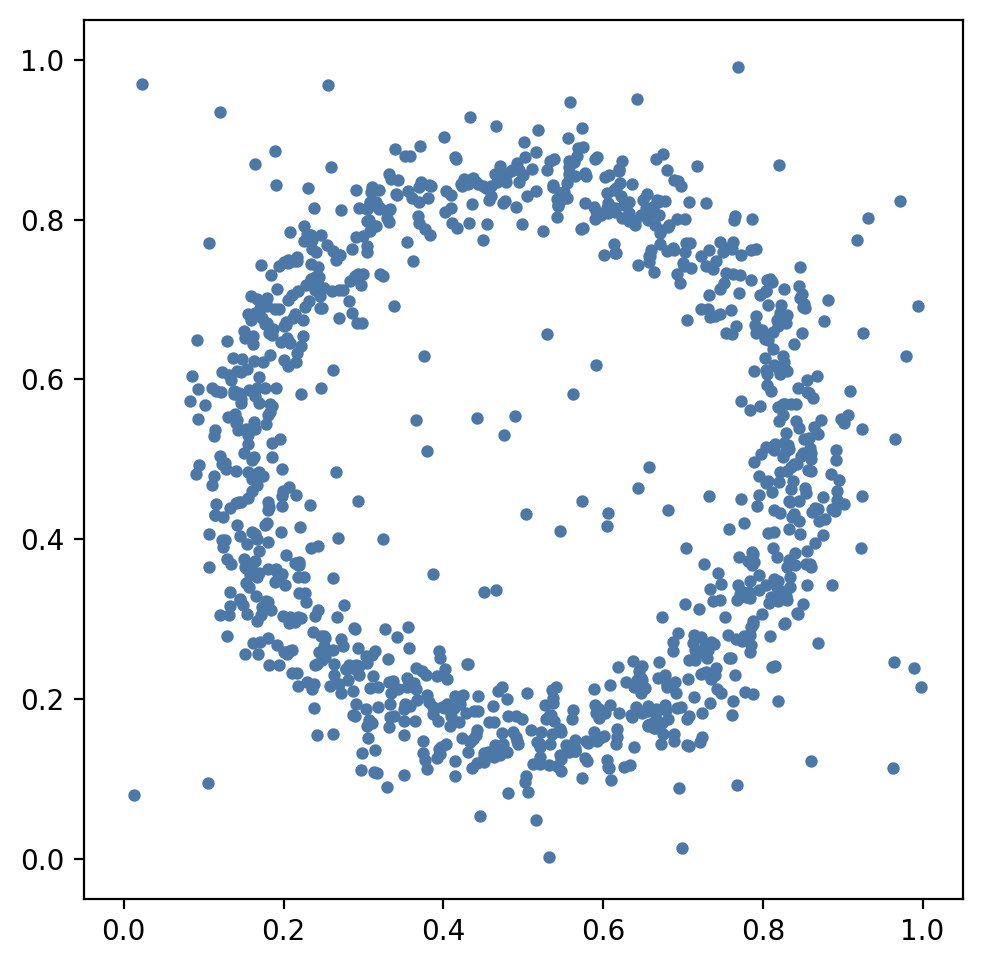}
    \caption{The first point cloud}\label{ExaData1}
  \end{subfigure}
  \quad
  \begin{subfigure}[b]{0.4\textwidth}
    \centering
    \includegraphics[width = 3.5cm]{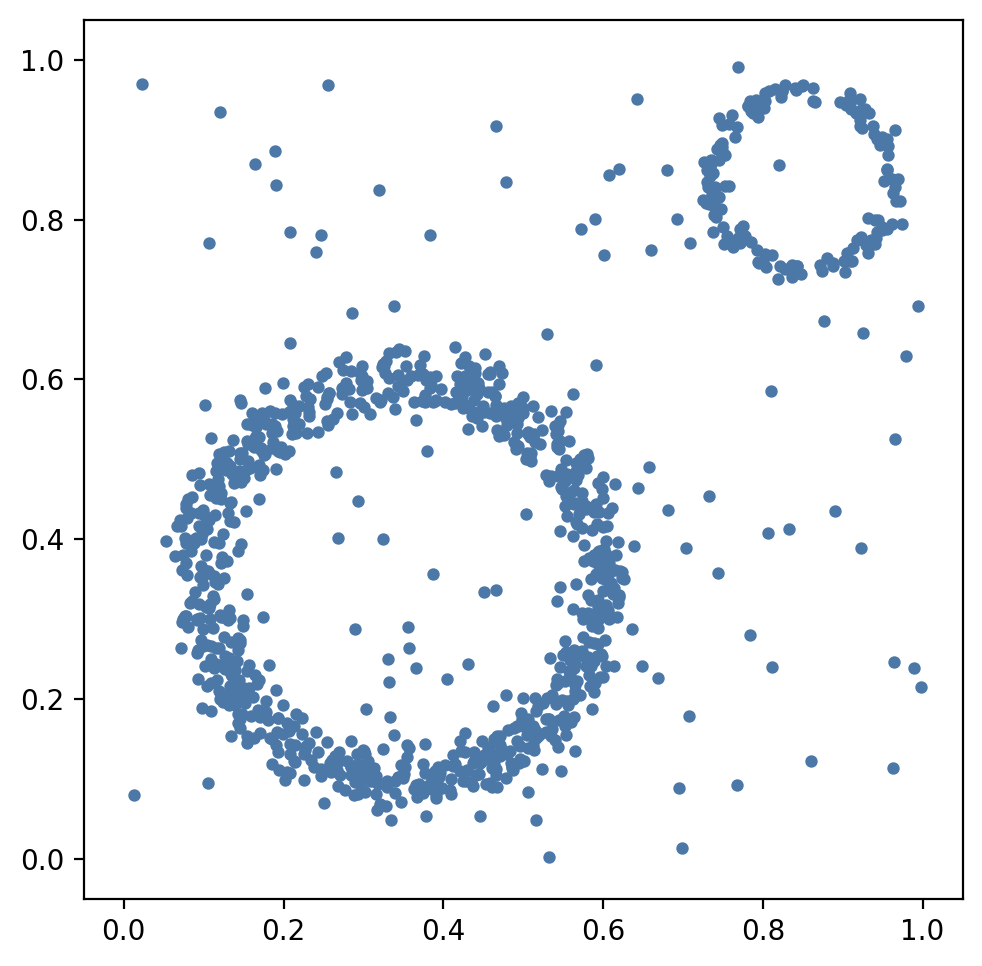}
    \caption{The second point cloud}\label{ExaData2}
  \end{subfigure}
  \caption{{\bf Point clouds with $r=0.1$.}}\label{ExaData}
\end{figure}


\subsection{\bf{Performance Comparison with CLA}}\label{ComCLA}

To evaluate how much more effective RCLA is at analyzing data than CLA, 
we quantitatively compare the two algorithms by measuring 
the bottleneck distance from the PD of each algorithm to that of $X_{\mathrm{shape}}$.

We use the first point cloud $X=X_{\mathrm{shape}}\cup X_{\mathrm{noise}}$ with $r=0.10$ introduced in Section~\ref{Data}. 
Figure~\ref{fig:comparison_cla} shows the graphs and the persistent diagrams of the original data $X$, the CLA result $X_{\delta}^{c}$, and the RCLA result $X_{\delta, k}^{c}$.
From the figure, we observe that the graph obtained by CLA still contains a noticeable noise part,
while the graph obtained by RCLA removes the noise effectively. 
\begin{figure}[h!]
\begin{tabular}{ @{} c @{} c @{} c @{} c @{} }
\multirow{2}{*}[0pt]{%
  \begin{tabular}{@{}c@{}}$X$\\[2pt]\includegraphics[width=0.2\linewidth]{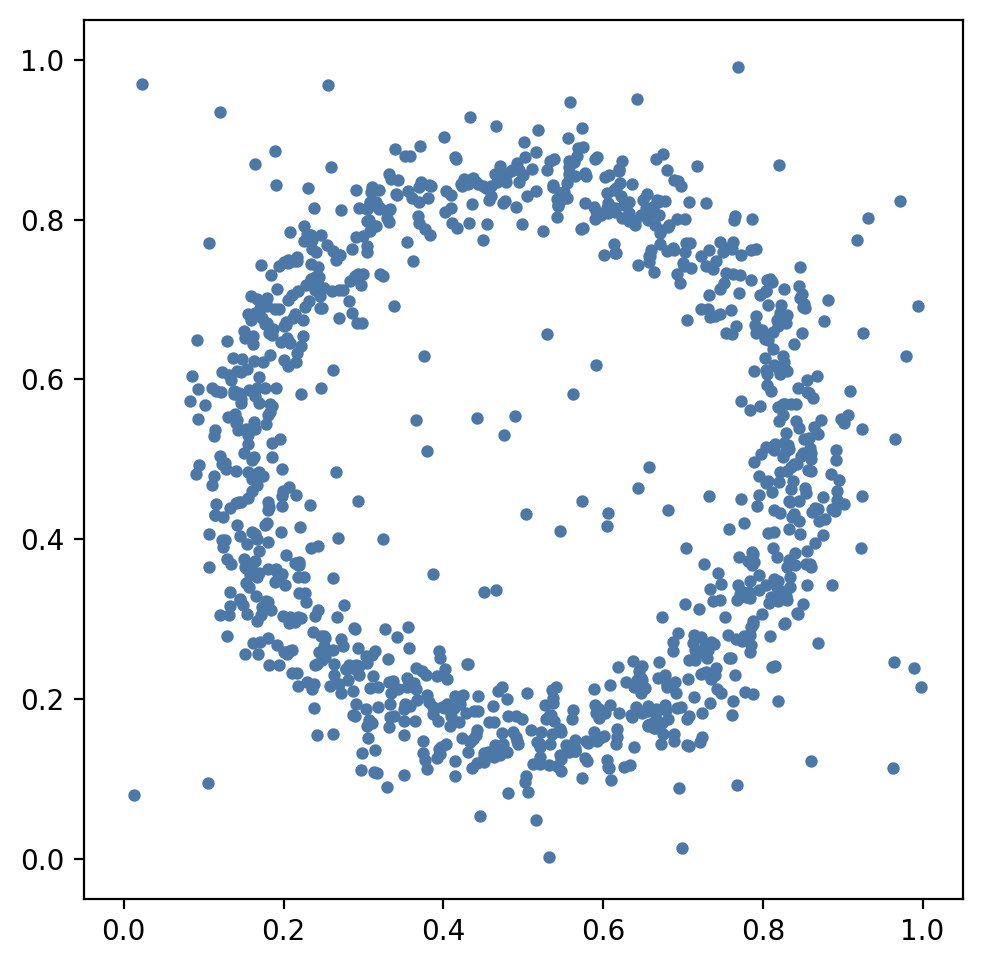}\end{tabular}%
}
&
\begin{tabular}{@{}c@{}}$X_{\mathrm{shape}}$\\[2pt]\includegraphics[width=0.2\linewidth]{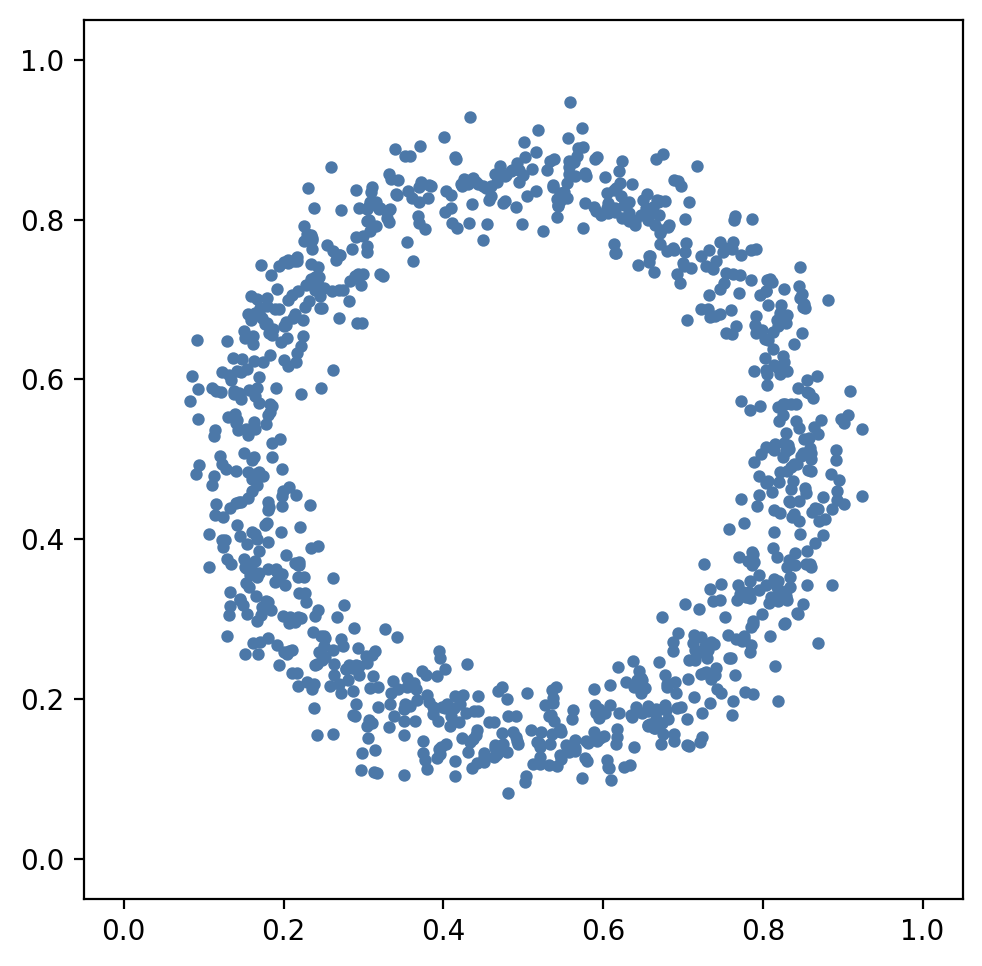}\end{tabular}
&
\begin{tabular}{@{}c@{}}$X_{\delta}^{c}$\\[2pt]\includegraphics[width=0.2\linewidth]{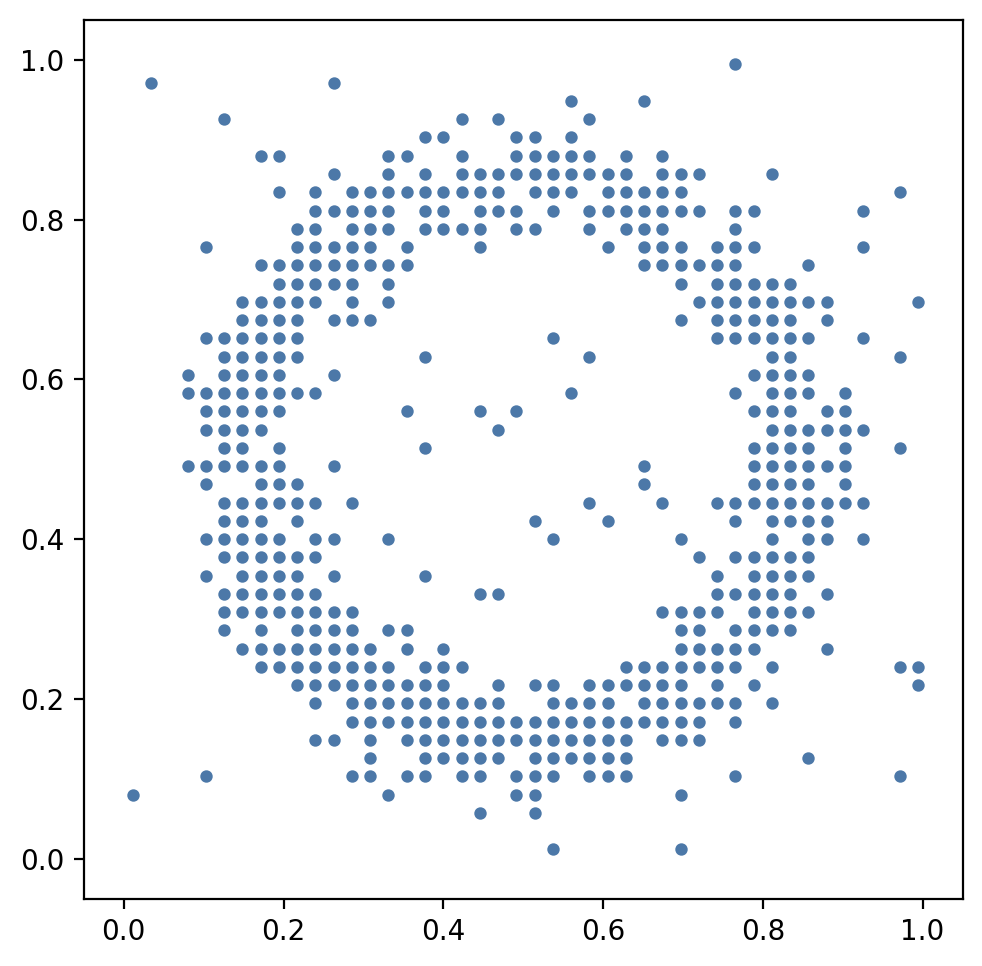}\end{tabular}
&
\begin{tabular}{@{}c@{}}$X_{\delta, k}^{c}$\\[2pt]\includegraphics[width=0.2\linewidth]{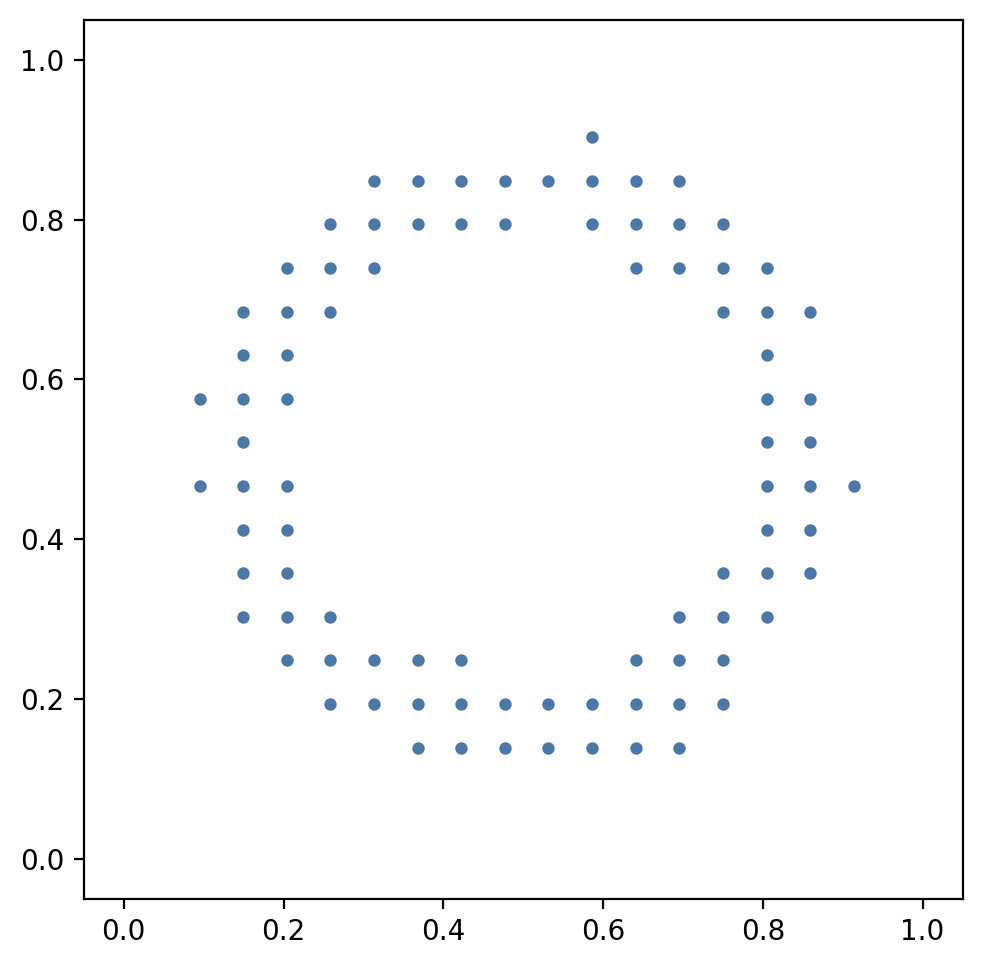}\end{tabular}
\\
&
\includegraphics[width=0.2\linewidth]{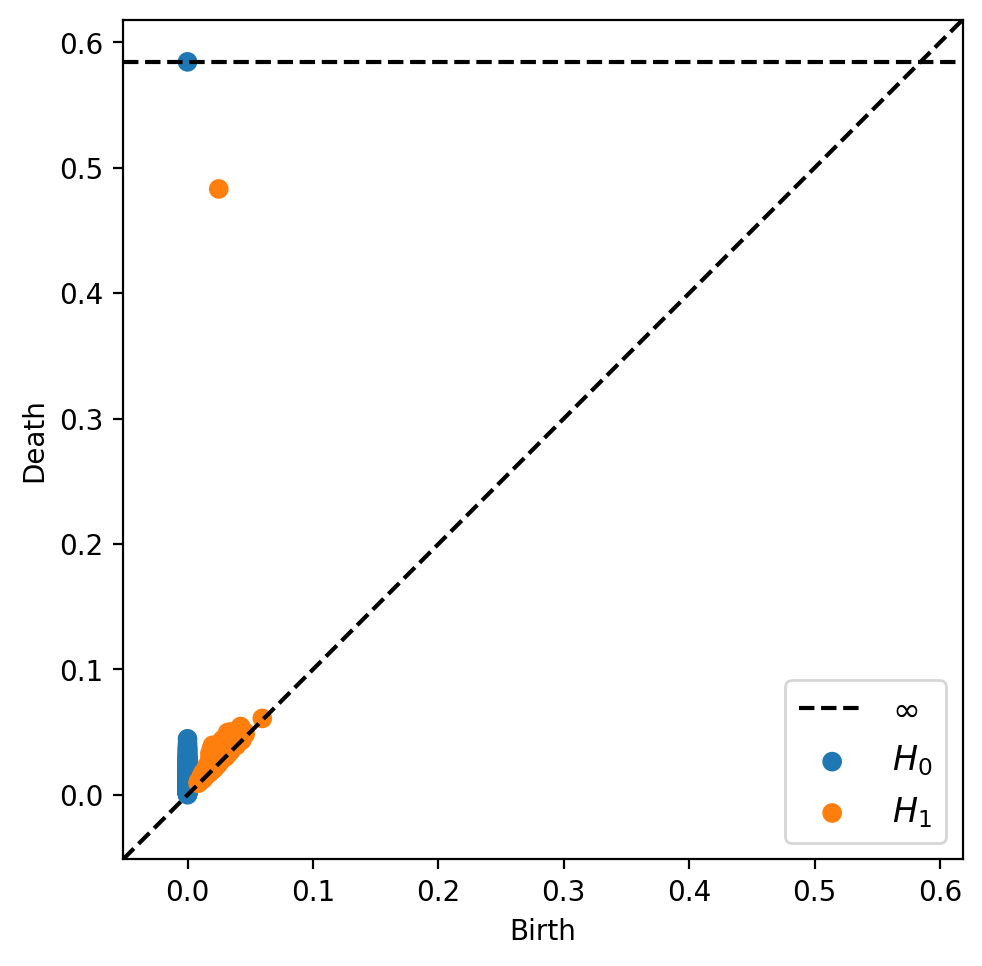}
&
\includegraphics[width=0.2\linewidth]{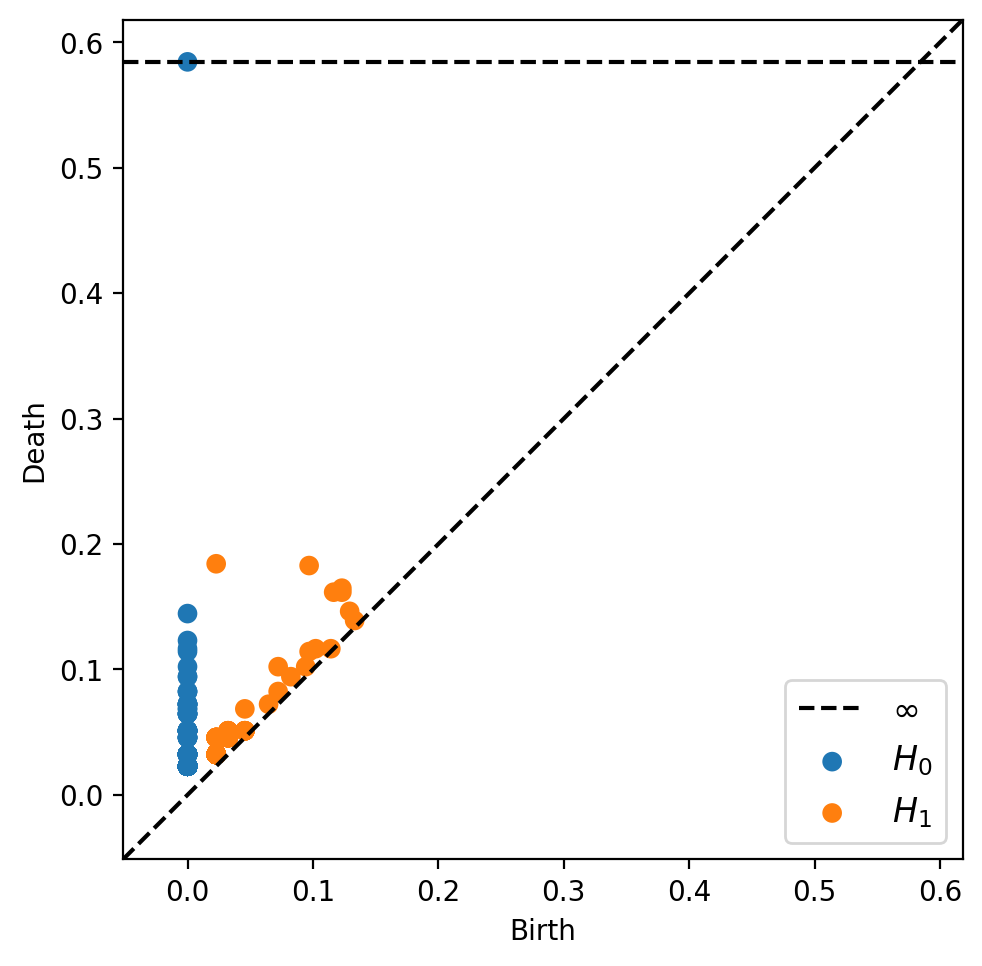}
&
\includegraphics[width=0.2\linewidth]{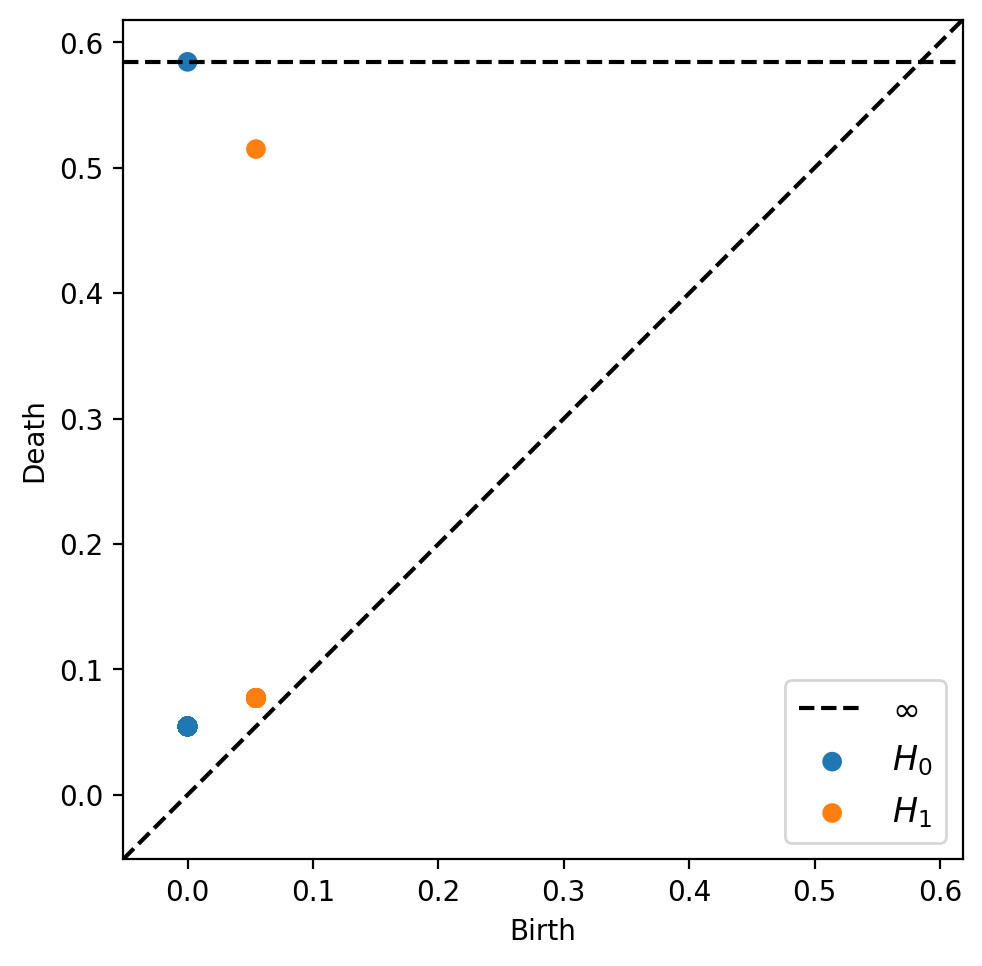}
\\
\end{tabular}
\caption{{\bf Comparison of point clouds and their persistence diagrams.}}
\label{fig:comparison_cla}
\end{figure}

This observation is quantitatively supported by the bottleneck distances to the $H_1$ persistence diagram of $X_{\mathrm{shape}}$ in Table~\ref{table:CLAvsRCLA}: 
the distance is $0.229107$ for CLA, whereas it is only $0.031811$ for RCLA.
Hence, RCLA preserves the topological features of $X_{\mathrm{shape}}$ significantly better than CLA.

\begin{table}[h]
\centering
{\setlength{\tabcolsep}{20pt}
\begin{tabular}{cc}
\toprule
Method & $d_B$ \\
\midrule
CLA  & 0.229107 \\
RCLA & 0.031811 \\
\bottomrule
\end{tabular}}
\caption{Bottleneck distances between the persistence diagrams of $X_{\mathrm{shape}}$ and those reconstructed by CLA and RCLA.}\label{table:CLAvsRCLA}
\end{table}




\subsection{\bf{Performance Comparison with Denoising Algorithms}}\label{ComDenoising}

We evaluate the denoising performance of RCLA against three established algorithms: Adaptive DBSCAN, LDOF, and LUNAR. 
The first comparison method, Adaptive DBSCAN~\cite{ZZLLL24}, is an improved version of the well-known DBSCAN algorithm. 
LDOF~\cite{ZHJ09} is a relatively recent denoising method based on the $k$-nearest neighbor algorithm, while LUNAR~\cite{GHNN22} is a representative machine learning-based denoising model.
These methods were employed in order to evaluate the effectiveness of RCLA in removing background noise and, at the same time, preserving the topological features of the underlying shape after denoising.

Using the second point cloud $X$ with $r\in\{0.10,0.15,0.20,0.25,0.30\}$ in Section~\ref{Data} where, we compare the ability of each algorithm to remove noise while preserving the underlying topological structure of the data. 
Performance is assessed via the bottleneck distance between the persistence diagram of the denoised output and that of the shape data $X_{\mathrm{shape}}$, where a smaller bottleneck distance indicates better preservation of topological features.

For each noise ratio $r$, the performance of RCLA, Adaptive DBSCAN, LDOF, and LUNAR was measured across $20$ random trials by computing the bottleneck distance between the $H_1$ persistence diagram of the output of each method and that of $X_{\mathrm{shape}}$.

Figure~\ref{fig:ComDenoising} shows, for each method, the best and worst outputs among these trials, selected by the smallest and largest bottleneck distances, respectively.
\begin{figure}[h!]
  \centering
  \begin{tabular}{>{\centering\arraybackslash}m{1.5cm}>{\centering\arraybackslash}m{2.5cm}>{\centering\arraybackslash}m{2.5cm}>{\centering\arraybackslash}m{2.5cm}>{\centering\arraybackslash}m{2.5cm}}
    & RCLA & Adaptive DBSCAN & LDOF & LUNAR \\
    \parbox{1.5cm}{\centering Best\\min($d_B$)} &
    \includegraphics[width=2.5cm]{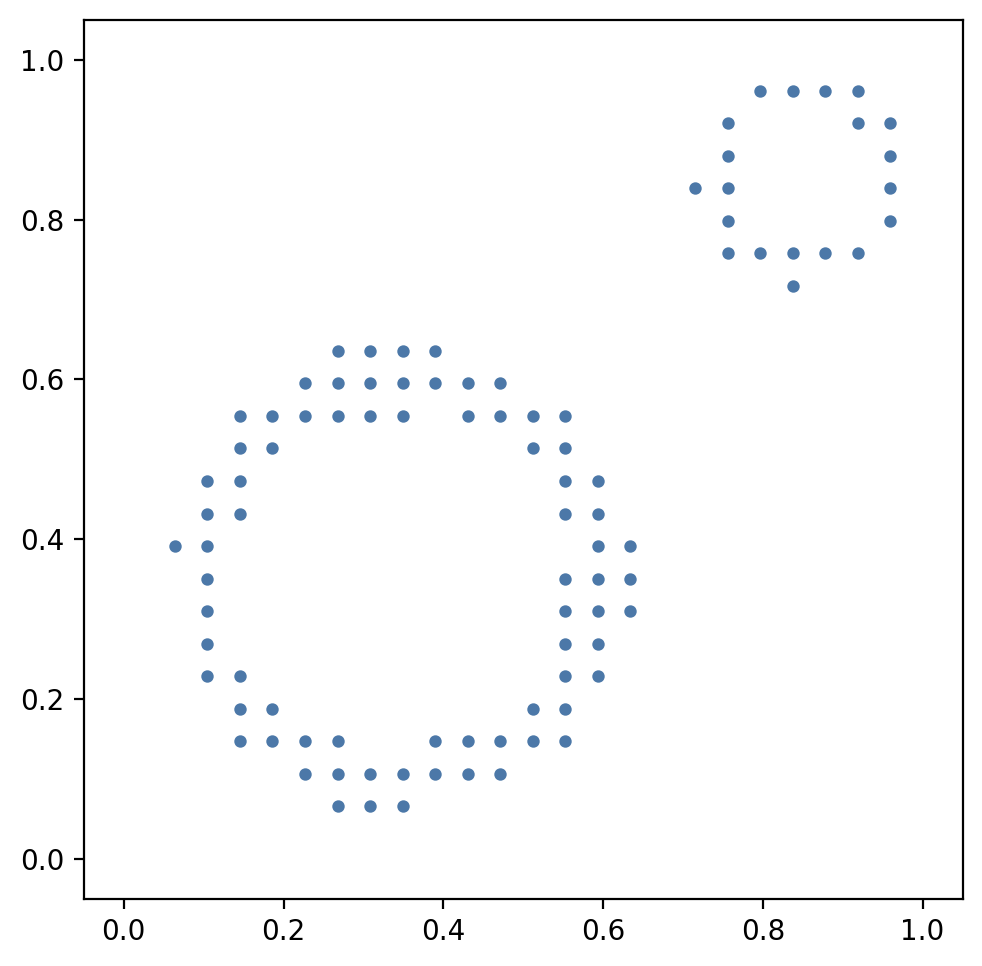} &
    \includegraphics[width=2.5cm]{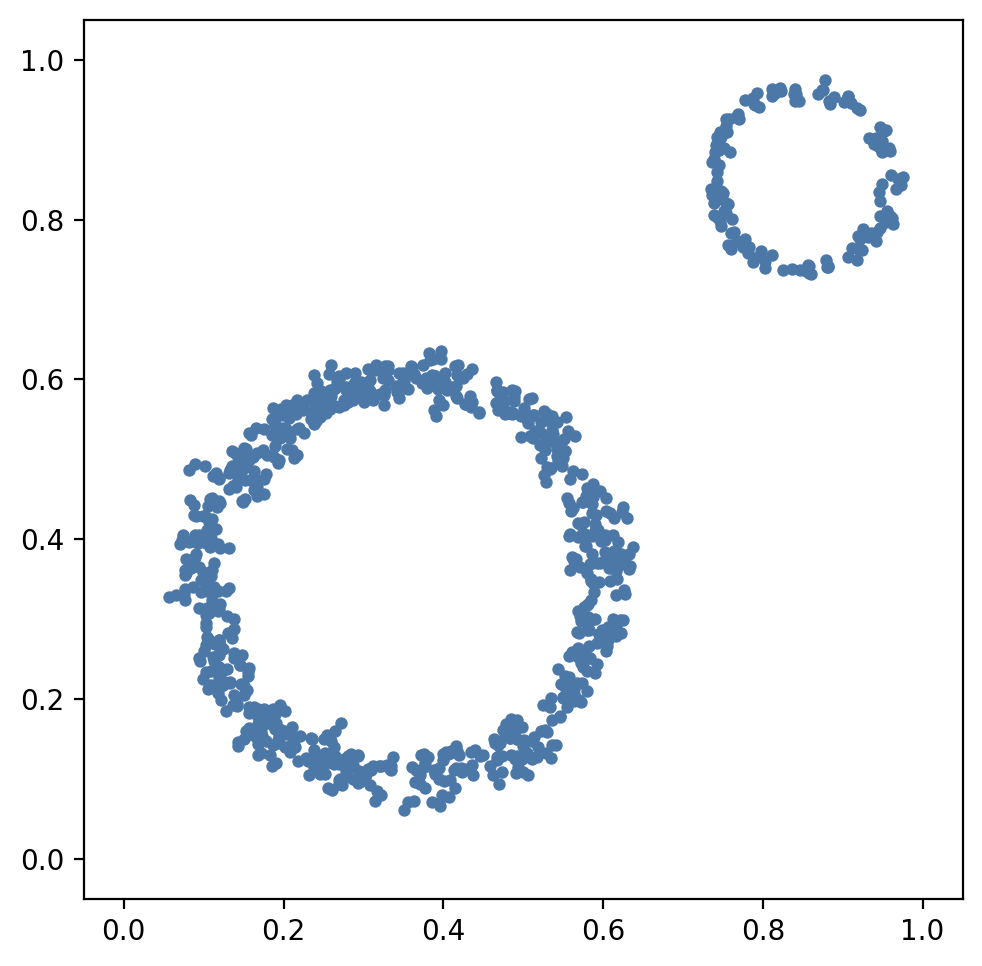} &
    \includegraphics[width=2.5cm]{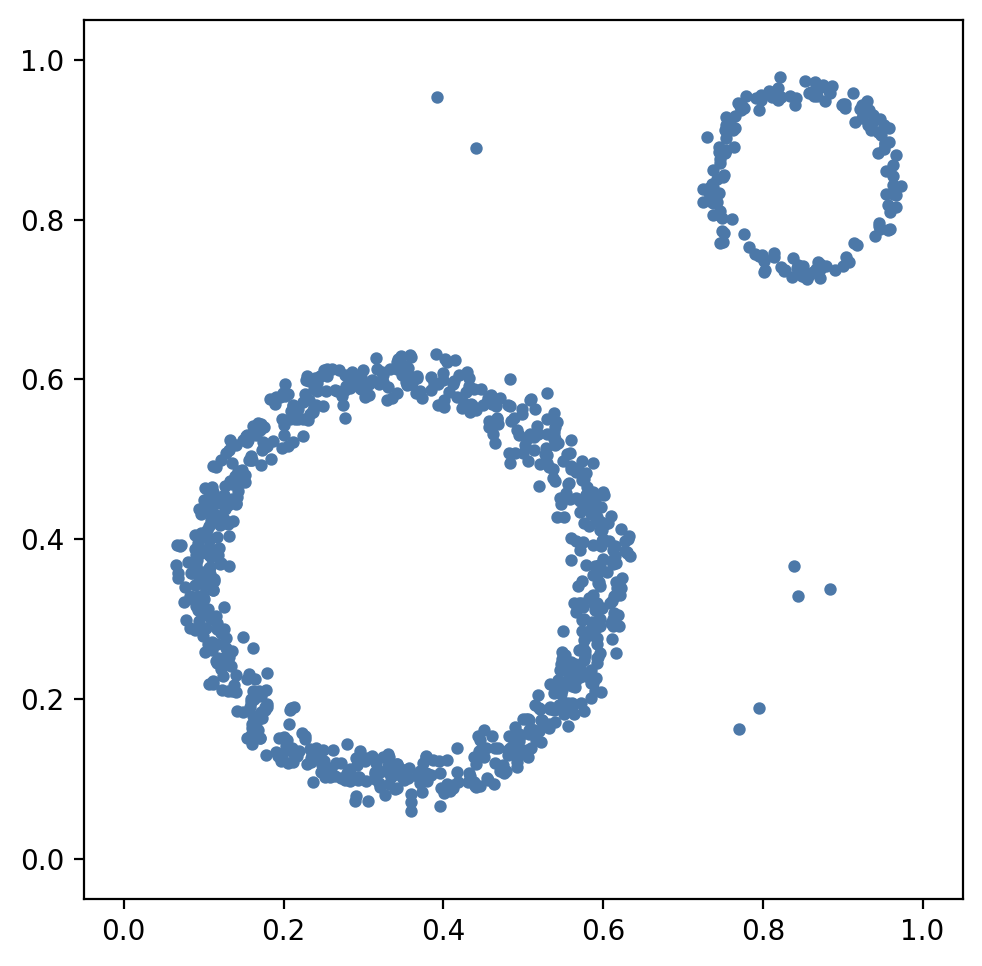} &
    \includegraphics[width=2.5cm]{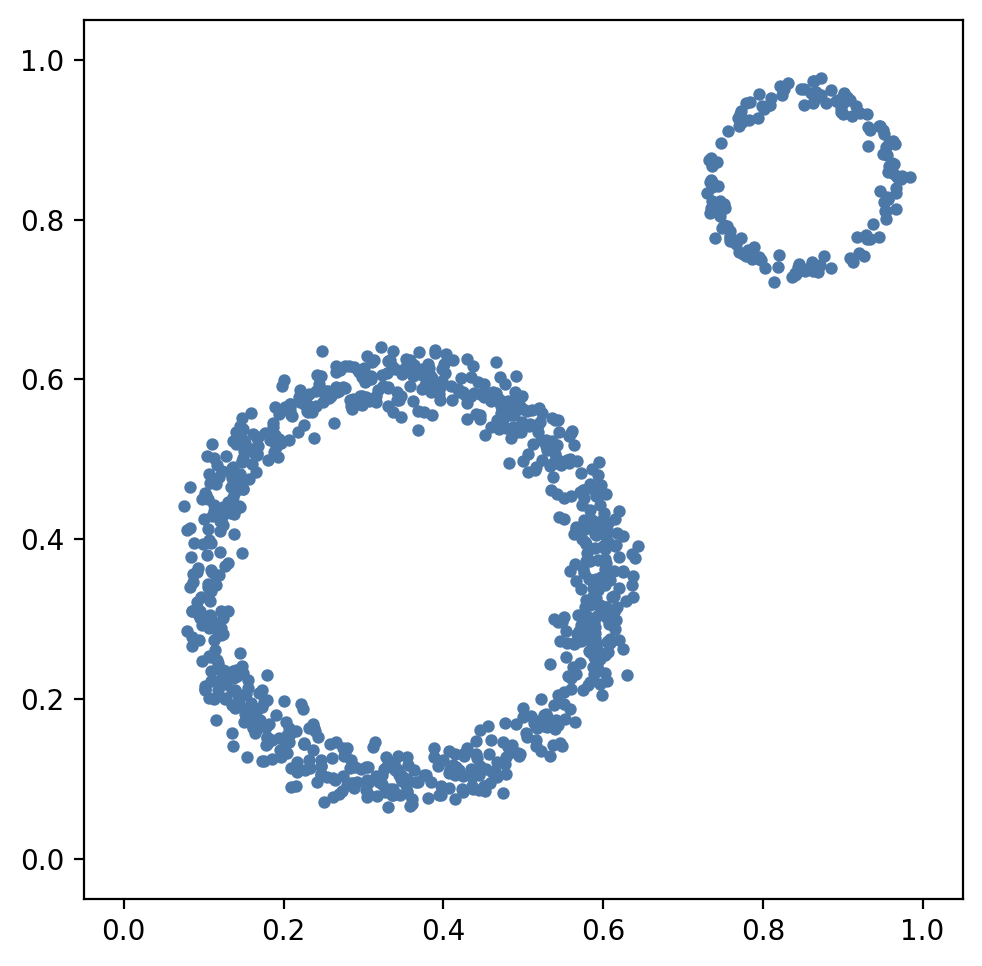} \\
    \parbox{1.5cm}{\centering Worst\\max($d_B$)} &
    \includegraphics[width=2.5cm]{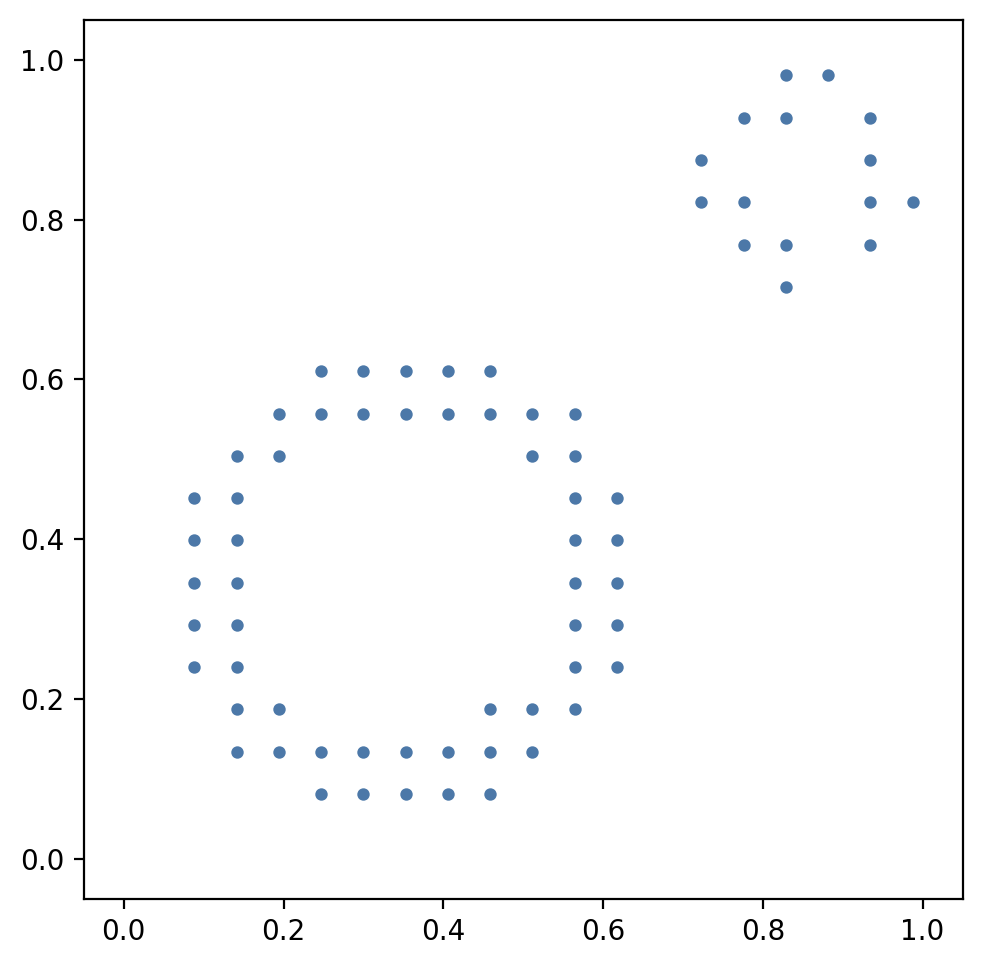} &
    \includegraphics[width=2.5cm]{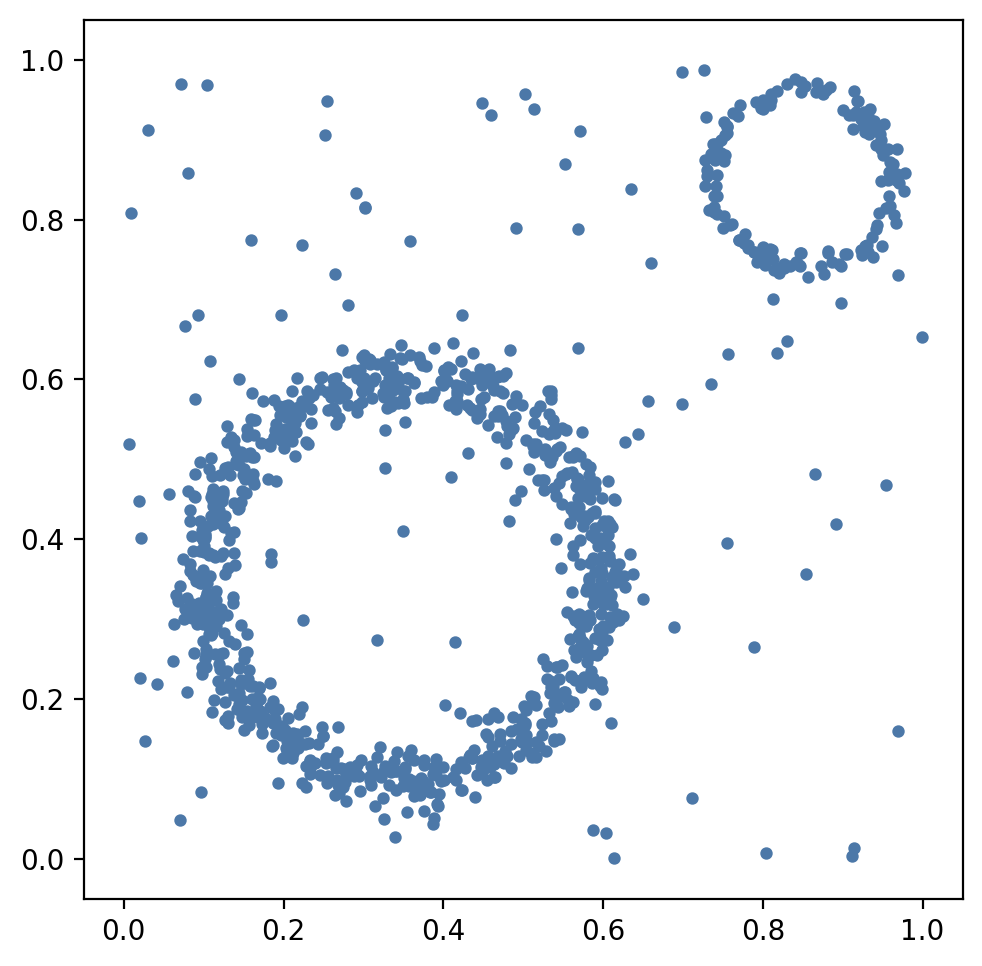} &
    \includegraphics[width=2.5cm]{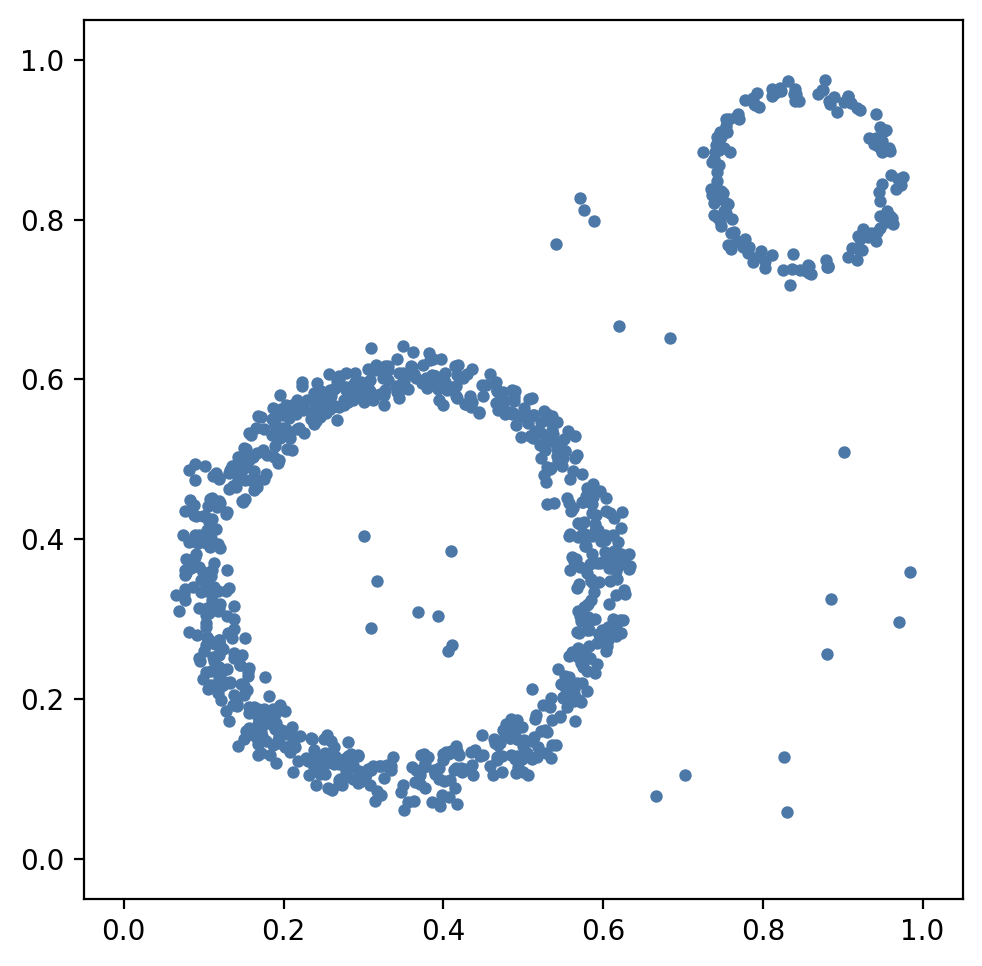} &
    \includegraphics[width=2.5cm]{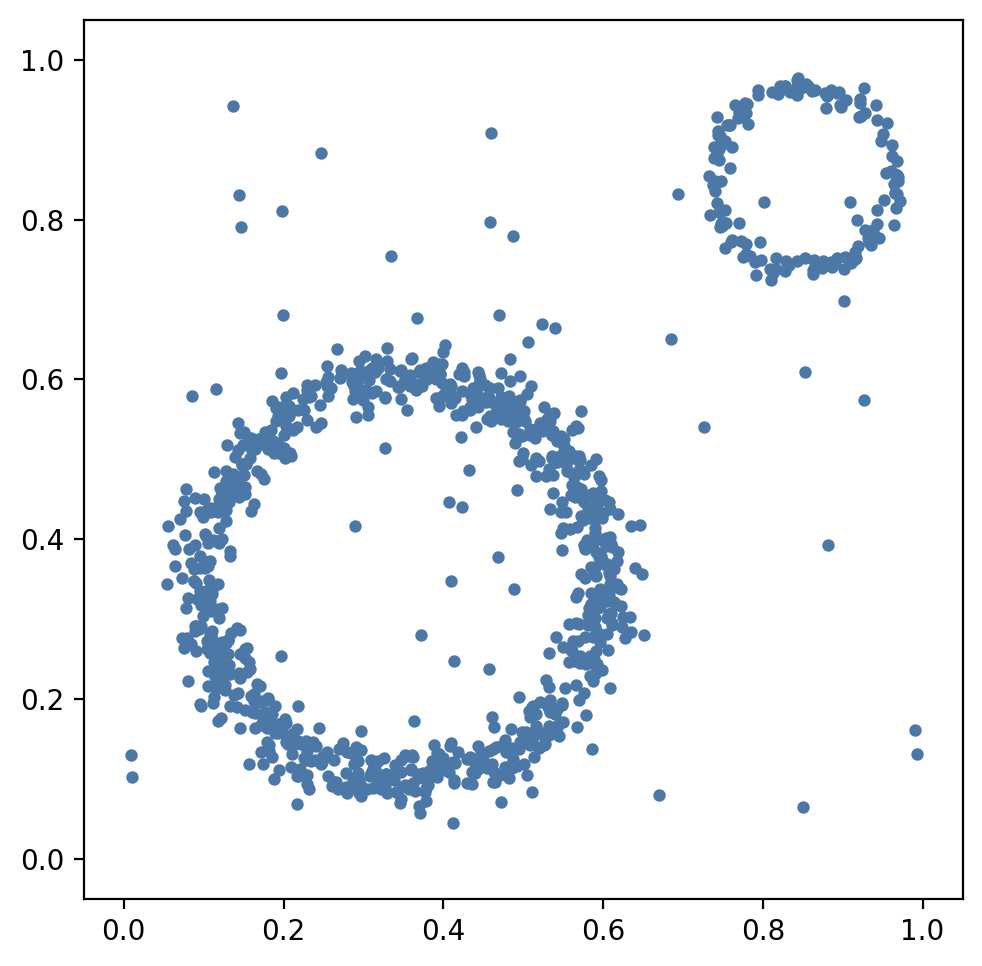} \\
  \end{tabular}
  \caption{\bf{Point clouds from each method, selected as those with the smallest and largest bottleneck distance from $\mathrm{PD}(X_{\text{shape}})$.}}\label{fig:ComDenoising}
\end{figure}
As seen in the figure, RCLA consistently preserves the topology of $X_{\mathrm{shape}}$ most faithfully, whereas the other three methods occasionally succeed in removing noise but often fail to do so.
To reflect this variability, the sample mean and standard deviation of the bottleneck distances are reported in Table~\ref{tab:denoise-comparison}.


\begin{table}[h]
\centering
\begin{tabular}{c|cccc}
\specialrule{0.08em}{0pt}{0pt}
\multirow{2}{*}{$r$} & \multicolumn{4}{c}{(Mean {\footnotesize $\pm$ SD})} \\
\cline{2-5}
& RCLA & Adaptive DBSCAN & LDOF & LUNAR \\
\midrule
0.10 
& 0.051069 {\footnotesize $\pm$ 0.016707}
& 0.090485 {\footnotesize $\pm$ 0.057214}
& 0.113232 {\footnotesize $\pm$ 0.037354}
& 0.054173 {\footnotesize $\pm$ 0.045463} \\

0.15 
& 0.052086 {\footnotesize $\pm$ 0.012052}
& 0.140838 {\footnotesize $\pm$ 0.043557}
& 0.136611 {\footnotesize $\pm$ 0.021376}
& 0.071958 {\footnotesize $\pm$ 0.045652} \\

0.20 
& 0.055251 {\footnotesize $\pm$ 0.012688}
& 0.141782 {\footnotesize $\pm$ 0.049153}
& 0.149163 {\footnotesize $\pm$ 0.021195}
& 0.079880 {\footnotesize $\pm$ 0.040987} \\

0.25 
& 0.051093 {\footnotesize $\pm$ 0.012029}
& 0.140576 {\footnotesize $\pm$ 0.048582}
& 0.159430 {\footnotesize $\pm$ 0.012704}
& 0.094894 {\footnotesize $\pm$ 0.038629}\\

0.30 
& 0.057440 {\footnotesize $\pm$ 0.013408}
& 0.154789 {\footnotesize $\pm$ 0.036262}
& 0.162117 {\footnotesize $\pm$ 0.008935}
& 0.134169 {\footnotesize $\pm$ 0.035352} \\
\specialrule{0.08em}{0pt}{0pt}
\end{tabular}
\caption{{\bf Mean and standard deviation of the bottleneck distance from $\mathrm{PD}(X_{\text{shape}})$ for each method.
}}
\label{tab:denoise-comparison}
\end{table}

As the noise ratio increases, the performance of all methods becomes degraded, and RCLA consistently yields smaller mean values than the other algorithms for all tested noise ratios. 
Moreover, RCLA exhibits a substantially smaller standard deviation than all the other algorithms in every experiment, indicating that its performance is more stable.

{\color{blue}

}


\vspace{0.5cm}
\section{Application to 3D Shape Classification}\label{Implementation}

The purpose of this experiment is to evaluate whether our algorithm can effectively capture the topological features of the data and perform accurate classification even in the presence of noise.

We use the 3D animal shape data from the deformation transfer dataset of Sumner and Popovi\'{c}~\cite{SumPop04}. 
In our experiments, we consider three animal classes: camel, elephant, and horse. 
For each class, $100$ OBJ meshes are randomly selected with replacement from the available pose variations.
From each mesh, a prescribed number $n$ of surface vertices is sampled without replacement. 

Each point cloud is translated so that it is centered in the unit cube $[0,1]^3$, as illustrated in the first row of Figure~\ref{fig:animal_data}, yielding three data classes: Camel, Elephant, and Horse.
To simulate noise contamination with noise ratio $r$, we add points drawn from an HPPP with intensity $\lambda = n \times r$ within $[0,1]^3$ to each point cloud, as depicted in the second row of Figure~\ref{fig:animal_data}, resulting in three additional classes: Noisy-Camel, Noisy-Elephant, and Noisy-Horse.
This gives a total of six data classes.

\begin{figure}[h!]
  \centering
  \includegraphics[width = 9cm]{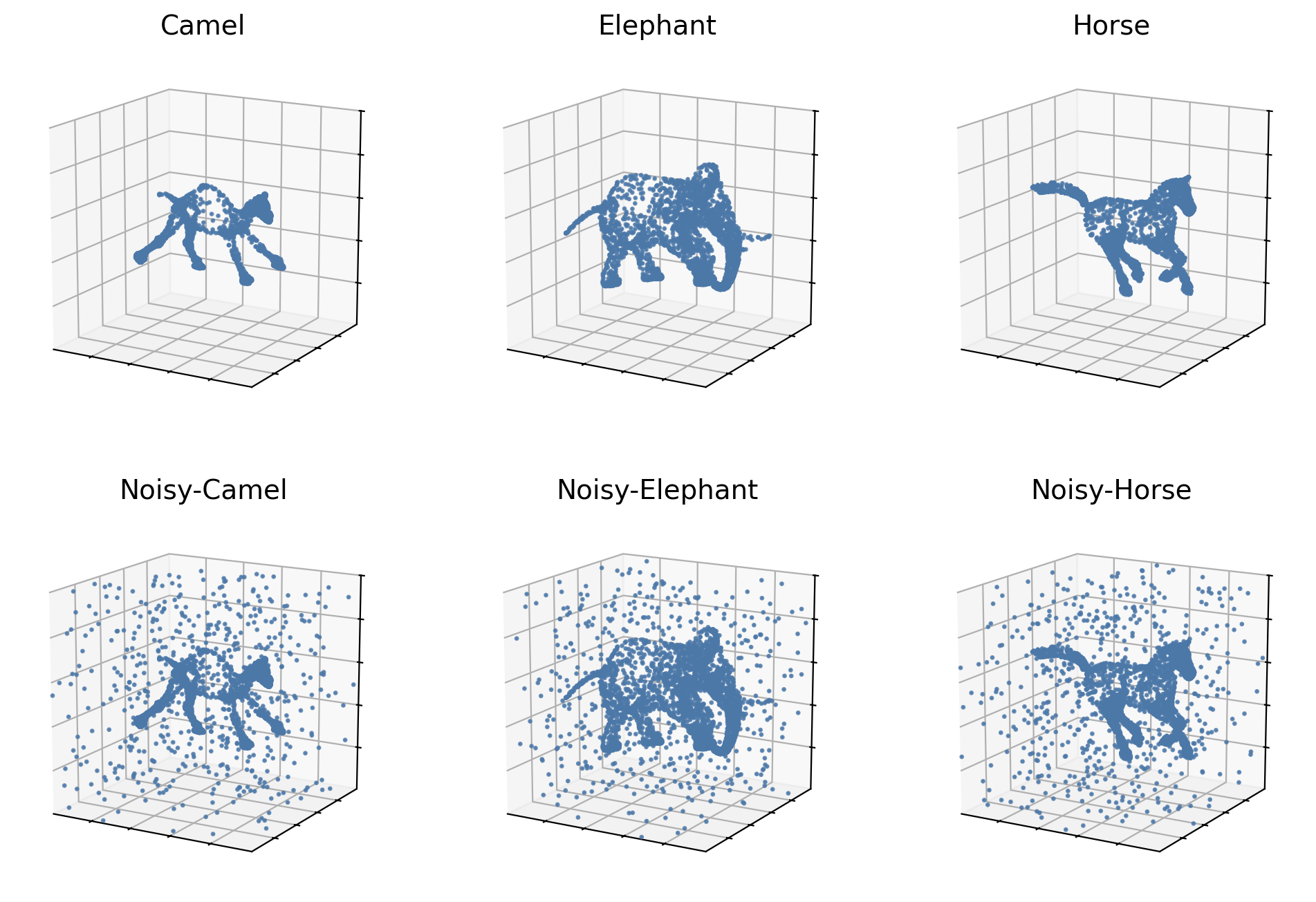}
  \caption{{\bf Six data classes of the 3D animal data with $n=2000$ and $r=0.2$}}\label{fig:animal_data}
\end{figure}

We vary two parameters: $n = 2000$, $3000$, $4000$, $5000$ and $r = 0.10$, $0.15$, $0.20$, $0.25$, $0.30$, yielding $20$ experimental settings.
For each class, $n$ points are randomly sampled from each OBJ mesh, resulting in a total of $100$ samples per class and $600$ samples in total.

Using these six data classes, we conducted classification experiments to evaluate whether RCLA can successfully distinguish the three animal classes: camel, elephant, and horse.
The experiment follows the pipeline illustrated in Figure~\ref{fig:pipeline}.

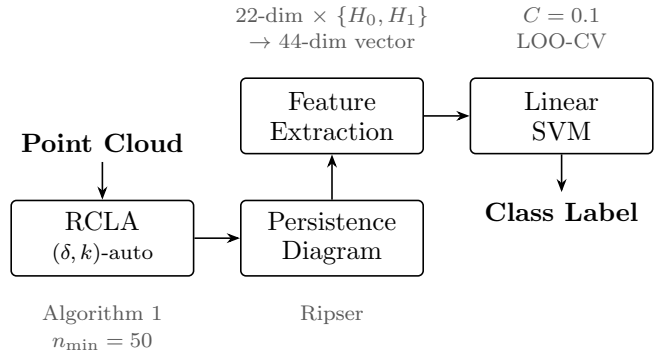
\begin{figure}[t]
\centering
\begin{tikzpicture}[
    node distance=0.6cm,
    box/.style={
        rectangle, draw, rounded corners=2pt,
        minimum height=1.0cm, minimum width=2.4cm,
        align=center, font=\small,
        fill=white, line width=0.6pt
    },
    arr/.style={-{Stealth[length=5pt]}, line width=0.6pt},
    label/.style={font=\scriptsize, text=gray!70!black, align=center}
]

\node[box] (rcla) {RCLA\\[-1pt]{\scriptsize$(\delta,k)$-auto}};
\node[box, right=of rcla] (ripser) {Persistence\\[-1pt]Diagram};
\node[box, above=of ripser] (feat) {Feature\\[-1pt]Extraction};
\node[box, right=of feat] (svm) {Linear\\[-1pt]SVM};

\draw[arr] (rcla) -- (ripser);
\draw[arr] (ripser) -- (feat);
\draw[arr] (feat) -- (svm);

\node[label, below=0.25cm of rcla] {Algorithm~\ref{alg:auto-select}\\$n_{\min}=50$};
\node[label, below=0.25cm of ripser] {Ripser};
\node[label, above=0.25cm of feat] {$22$-dim $\times$ $\{H_0,H_1\}$\\$\to$ $44$-dim vector};
\node[label, above=0.25cm of svm] {$C=0.1$\\LOO-CV};

\node[font=\small, above=0.5cm of rcla] (input) {\bf{Point Cloud}};
\draw[arr] (input) -- (rcla);
\node[font=\small, below=0.5cm of svm] (output) {\bf{Class Label}};
\draw[arr] (svm) -- (output);

\end{tikzpicture}
\caption{{\bf The pipeline for classification.}}
\label{fig:pipeline}
\end{figure}

In the pipeline, a point cloud is reduced via RCLA with automatic $(\delta,k)$ selection, and a persistence diagram is computed with Ripser~\cite{Bau21}. For each of $H_0$ and $H_1$, we extract $22$ descriptive statistics---mean, standard deviation, min, max, and quartiles ($Q_{25}$, $Q_{50}$, $Q_{75}$) of birth~$b$, death~$d$, and lifetime~$\ell=d{-}b$, plus the point count and total persistence~$\sum\ell$---yielding a $44$-dimensional feature vector~\cite{ali2023survey} that is classified by a linear Support Vector Machine (SVM) with leave-one-out cross-validation.

Table~\ref{tab:setting-acc} reports the classification accuracy for each configuration.
For all $20$ settings, the accuracy remained above $99.50\%$, and $9$ out of $20$ achieved $100.00\%$.
The overall mean accuracy is $99.88\%$.

\begin{table}[htbp]
\centering
\begin{tabular}{cccccc}
\toprule
$n$ & $r=0.10$ & $r=0.15$ & $r=0.20$ & $r=0.25$ & $r=0.30$ \\
\midrule
2000 & 99.83 & 100.00 & 99.67 & 100.00 & 99.50 \\
3000 & 99.83 & 100.00 & 100.00 & 100.00 & 100.00 \\
4000 & 99.67 & 100.00 & 99.83 & 100.00 & 100.00 \\
5000 & 99.83 & 99.83 & 100.00 & 99.83 & 99.83 \\
\bottomrule
\end{tabular}
\caption{{\bf Classification accuracy ($\%$).}}
\label{tab:setting-acc}
\end{table}

\vspace{0.5cm}
\section{Conclusion}\label{Conclusion}

We propose RCLA, an improved version of CLA, which simultaneously performs data reduction and noise removal in point clouds. 
In contrast to CLA, which was originally designed for data reduction and is therefore sensitive to background noise, RCLA effectively handles both tasks in a single framework.
The stability theorem for RCLA provides a theoretical justification for the preservation of topological features in the presence of HPPP noise.
The HPPP is commonly used to model background noise that is uniformly distributed over a spatial domain.
We devise Algorithm~\ref{alg:auto-select}, which effectively determines $\delta$ and $k$ for data with such background noise.
With this automatic parameter selection, RCLA can be applied to all experimental settings in this paper. 

We first compare RCLA with CLA to examine the effect of the proposed refinement using the point cloud (1) of Figure~\ref{ExaData}.
As shown in Table~\ref{table:CLAvsRCLA}, the bottleneck distance to the $H_1$ persistence diagram of $X_{\mathrm{shape}}$ drops from $0.229107$ for CLA to just $0.031811$ for RCLA, demonstrating that RCLA is far more effective at preserving the topological structure of the underlying shape.

As the RCLA method also removes background noise, we compared its performance with several denoising algorithms of different types; the results are shown in Table~\ref{tab:denoise-comparison}.
A smaller mean bottleneck distance indicates that the output is topologically closer to $X_{\mathrm{shape}}$, while a smaller standard deviation indicates greater stability across trials. 
In our experiments, RCLA shows both the smallest mean bottleneck distance and the smallest standard deviation over all noise configurations. 
By contrast, Adaptive DBSCAN, LDOF, and LUNAR yield larger bottleneck distances and greater variability, which is also illustrated in Figure~\ref{fig:ComDenoising}. 
These results indicate that RCLA recovers the topological structure of the underlying shape more effectively and more consistently than the other methods.


As an application, we construct a pipeline based on RCLA for shape classification of three-dimensional dataset, including both noise-free and noisy classes.
The pipeline includes the vectorization of topological summaries and the subsequent machine learning step.
Table~\ref{tab:setting-acc} reports almost perfect classification accuracy across all $20$ settings, with a mean accuracy of $99.88\%$ ranging from $99.50\%$ to $100.00\%$.
The results demonstrate that RCLA effectively preserves the necessary topological features by removing noise, thereby distinguishing between the three animal classes even in noisy environments.





\section*{\bf{Acknowledgement}}

All authors equally contribute this paper. 
The work of Seung Yeop Yang was supported by the National Research Foundation of Korea(NRF) grant funded by the Korea government(MSIT) (RS-2022-NR070839).
The work of Semin Oh and Seung Yeop Yang was supported by Global - Learning \& Academic research institution for Master’s·PhD students, and Postdocs (G-LAMP) Program of the National Research Foundation of Korea (NRF) grant funded by the Ministry of Education (No. RS-2023-00301914).
The work of Seonmi Choi was supported by the Basic Science Research Program through the National Research Foundation of Korea(NRF) funded by the Ministry of Education (No. RS-2021-NR065036).
The work of Jeong Rye Park was supported by the Basic Science Research Program through the National Research Foundation of Korea(NRF) funded by the Ministry of Education (No. RS-2021-NR065455).

\end{document}